\documentclass[aps,prd,twocolumn,nofootinbib,superscriptaddress,preprintnumbers]{revtex4-1}
\usepackage{amsfonts,amsmath,amssymb,amsthm,braket}
\usepackage{booktabs,array}
\usepackage[utf8]{inputenc}
\usepackage{float,graphicx,hhline}
\usepackage{algorithm,algpseudocode}
\usepackage{appendix}
\usepackage[dvipsnames]{xcolor}
\usepackage{mathtools}
\usepackage{multirow}
\usepackage{tikz-cd} 

\AtBeginDocument{
\heavyrulewidth=.08em
\lightrulewidth=.05em
\cmidrulewidth=.03em
\belowrulesep=.65ex
\belowbottomsep=0pt
\aboverulesep=.4ex
\abovetopsep=0pt
\cmidrulesep=\doublerulesep
\cmidrulekern=.5em
\defaultaddspace=.5em
}

\usepackage{hyperref}
\hypersetup{
    colorlinks=true,     
    linkcolor=blue,      
    citecolor=blue,      
    filecolor=blue,      
    urlcolor=blue        
}
\usepackage{cleveref}

\newcommand{\D}[1]{\mathcal{D}{#1}\;}
\newcommand{\DU}[0]{\mathcal{D}{U}}
\newcommand{\DV}[0]{\mathcal{D}{V}}
\newcommand{\obs}[0]{\mathcal{O}}
\newcommand{\G}[0]{G}
\newcommand{\GG}[0]{\mathcal{G}}
\newcommand{\Seff}[0]{S_{\text{eff}}}

\newcommand{\ESS}[0]{\text{ESS}}
\newcommand{\SU}[1]{\textrm{SU}(#1)}
\newcommand{\SUn}{\SU{N}}
\newcommand{\U}[1]{\textrm{U}(#1)}
\newcommand{\Un}{\U{N}}
\newcommand{\modelparams}{\xi}
\newcommand{\boxflow}{\chi}
\newcommand{\LDJ}{\textrm{LDJ}}
\newcommand{\torusalg}{\mathfrak{t}}

\let\Re\undefined
\let\Im\undefined
\DeclareMathOperator{\Re}{Re}
\DeclareMathOperator{\Im}{Im}
\DeclareMathOperator{\DKL}{D_{\text{KL}}}
\DeclareMathOperator{\tr}{tr}

\DeclareMathOperator{\diag}{diag}

\newtheorem{prop}{Proposition}
\newtheorem{lemma}{Lemma}
\newtheorem{ex}{Example}

\newenvironment{proof}{\noindent\textbf{Proof\ }}{\hspace*{\fill}$\Box$\medskip}

\newtheorem{coro}{Corollary}

\begin{document}

\author{Denis~Boyda}
 \email{boyda@mit.edu}
\affiliation{Center for Theoretical Physics, Massachusetts Institute of Technology, Cambridge, MA 02139, USA}
\author{Gurtej~Kanwar}
 \email{gurtej@mit.edu}
\affiliation{Center for Theoretical Physics, Massachusetts Institute of Technology, Cambridge, MA 02139, USA}
\author{S\'{e}bastien~Racani\`{e}re}
 \email{sracaniere@google.com}
\affiliation{DeepMind, London, UK}
\author{Danilo~Jimenez~Rezende}
 \email{danilor@google.com}
\affiliation{DeepMind, London, UK}
\author{Michael~S.~Albergo}
\affiliation{Center for Cosmology and Particle Physics,
New York University, New York, NY 10003, USA}
\author{Kyle~Cranmer}
\affiliation{Center for Cosmology and Particle Physics,
New York University, New York, NY 10003, USA}
\author{Daniel~C.~Hackett}
\affiliation{Center for Theoretical Physics, Massachusetts Institute of Technology, Cambridge, MA 02139, USA}
\author{Phiala~E.~Shanahan}
\affiliation{Center for Theoretical Physics, Massachusetts Institute of Technology, Cambridge, MA 02139, USA}

\title{Sampling using $\SUn$ gauge equivariant flows}

\date{\today}
\preprint{MIT-CTP/5228}

\begin{abstract}
We develop a flow-based sampling algorithm for $\SUn$ lattice gauge theories that is gauge-invariant by construction. Our key contribution is constructing a class of flows on an $\SUn$ variable (or on a $\Un$ variable by a simple alternative) that respect matrix conjugation symmetry. We apply this technique to sample distributions of single $\SUn$ variables and to construct flow-based samplers for $\SU{2}$ and $\SU{3}$ lattice gauge theory in two dimensions.
\end{abstract}

\maketitle

\section{Introduction}

Gauge theories based on $\SUn$ or $\Un$ groups describe many aspects of nature. For example, the Standard Model of nuclear and particle physics is a non-abelian gauge theory with the symmetry group $\U{1} \times \SU{2} \times \SU{3}$, candidate theories for physics beyond the Standard Model can be defined based on strongly interacting $\SUn$ gauge theories~\cite{Kribs:2016cew,DeGrand:2015zxa}, $\SUn$ gauge symmetries emerge in various condensed matter systems~\cite{Ichinose:2014cba,Sachdev2018EmergentGauge,Chatterjee:2019xgs,Guo2020EmergentGauge,Bi:2018xvr}, and $\SUn$ and $\Un$ gauge symmetries feature in the low energy limit of certain string-theory vacua~\cite{Giveon:1998sr}. In the context of the rapidly-developing area of machine-learning applications to physics problems, the incorporation of gauge symmetries in machine learning architectures is thus of particular interest~\cite{cohen2019gauge,bekkers2019bspline,Kanwar:2020xzo,lezcanocasado2019cheap,lezcanocasado2019trivializations,pmlr-v89-falorsi19a}.

Here, we demonstrate how $\SUn$ gauge symmetries can be incorporated into \emph{flow-based models}~\cite{papamakarios2019normalizing}. These models use a parameterized invertible transformation (a ``flow'') to construct a variational ansatz for a target probability distribution that can be optimized via machine learning techniques to enable efficient sampling. We detail the application of this approach to lattice field theory calculations, for which such samplers have been found to offer potentially significant advantages over more traditional sampling algorithms~\cite{Albergo:2019eim,Kanwar:2020xzo,Nicoli:2020njz}.

A general approach to incorporating a symmetry in flow-based sampling models is to construct the models in terms of invertible transformations that are \emph{equivariant} to symmetry operations, meaning that the transformation and symmetry operations commute. For any gauge theory with a continuous gauge group, we showed in Ref.~\cite{Kanwar:2020xzo} that a gauge equivariant transformation that simultaneously remains equivariant under a large subgroup of spacetime translations can be constructed in terms of a \emph{kernel}: a transformation that acts on elements of the gauge group and is equivariant under matrix conjugation, $U \rightarrow X U X^{-1}$, where $U$ and $X$ are elements of the gauge group in the fundamental matrix representation. In Ref.~\cite{Kanwar:2020xzo}, this approach was demonstrated in the context of $\U{1}$ gauge theory. Here, we develop a class of kernels for $\SUn$ group elements (and describe a similar construction for $\Un$ group elements). We show that if an invertible transformation acts only on the eigenvalues of a matrix and is equivariant under permutation of those eigenvalues, then it is equivariant under matrix conjugation and may be used as a kernel. Moreover, by making a connection to the maximal torus within the group and to the Weyl group of the root system, we show that this is in fact a universal way to define a kernel for unitary groups.

The application of flow-based models to lattice field theory is reviewed briefly in Sec.~\ref{subsec:flow-based-review}. Methods to impose symmetries in these models are reviewed in Sec.~\ref{subsec:flow-based-symm}, and Sec.~\ref{subsec:gauge-equiv} describes our particular approach to imposing gauge symmetry in flow-based models using single-variable kernels. In Sec.~\ref{sec:single-sun-flows}, we construct kernels for $\SUn$ variables and investigate sampling from distributions over such variables, including the marginal distributions relevant for plaquettes in two-dimensional lattice gauge theory. Finally, in Sec.~\ref{sec:sun-gauge-theory} we use these kernels to construct gauge-symmetric flow-based samplers for $\SU{2}$ and $\SU{3}$ lattice gauge theory in two dimensions, and demonstrate that observables in these theories are exactly reproduced by the flow-based sampling approach.

\section{Flow-based sampling for lattice gauge theory} 

Lattice quantum field theory provides a non-perturbative regularization of the path integral by discretizing the theory onto a spacetime lattice. In Euclidean spacetime, the regularized expectation value of an observable $\obs$ is defined in terms of the discretized action $S(U)$ by
\begin{equation}
\begin{gathered}\label{eq:vacevalO}
  \braket{\obs} = \frac{1}{Z} \int \D{U} \obs(U) e^{-S(U)},
  \quad Z = \int \D{U} e^{-S(U)},
\end{gathered}
\end{equation}
where $\int \DU$ integrates over all degrees of freedom of the discretized quantum field $U$. We denote by $U_{\mu}(x) \in \G$ the element of $U$ on link $(x, x+\hat{\mu})$, where $\mu \in \{1, 2, \dots, D\}$ is the spacetime direction of the link, $x \in \mathbb{Z}^{D}$ indicates a site on the $D$-dimensional spacetime lattice, and $\G$ is the structure group of the gauge theory; for many relevant physical theories, the structure groups are Lie groups. The path integral measure $\DU$ for a lattice gauge theory is a product of the Haar measure of $\G$ per link.

Eq.~\eqref{eq:vacevalO} can be evaluated numerically by sampling {\it configurations} from the probability distribution $p(U) = e^{-S(U)} / Z$, which is typically undertaken using Markov chain methods~\cite{Morningstar:2007zm}. In Refs.~\cite{Albergo:2019eim,Kanwar:2020xzo}, we developed an approach to evaluate Eq.~\eqref{eq:vacevalO} for lattice field theories by sampling independent configurations from a flow-based model optimized to approximate $p(U)$, where unbiased estimates of observables can be obtained from this approximate distribution by either a reweighting technique or a Metropolis accept/reject step.\footnote{Sampling for lattice field theories based on generative adversarial networks has also been investigated in related work~\cite{Urban:2018tqv,Zhou:2018ill}.} Flow-based methods can similarly be applied to statistical theories (with continuous degrees of freedom) by replacing the field configurations $U$ of Eq.~\eqref{eq:vacevalO} with microstates, replacing the action with the Hamiltonian over temperature, $S \rightarrow H/k_B T$, and interpreting the distribution as the Boltzmann distribution~\cite{LiWang2018NNRG,ZhangEWang2018Monge,noe2019boltzmann,li2019neural,Hartnett:2020}.

\subsection{Sampling gauge configurations using flows} \label{subsec:flow-based-review}

A flow-based sampler consists of two components:
\begin{enumerate}
\item A \emph{prior distribution}\footnote{We specify the distribution using a density function $r(V)$. Here and in the following, this is implicitly a density with respect to the path integral measure $\DV$ (or $\DU$).} $r(V)$ that is easily sampled;
\item An invertible function, or \emph{flow}, $f$ that has a tractable Jacobian factor.
\end{enumerate}
Here, we restrict discussion to flow-based models targeting distributions $p(U)$ on Lie groups $\GG$, for which $U \in \GG$ and $f: \GG \rightarrow \GG$. The group could be a product of structure groups $\GG = \G \otimes \G \otimes \dots$, as in the case of lattice gauge theory, or an unfactorizable group such as $\SUn$ or $\Un$.
Generating a sample from the model proceeds by first sampling from the prior distribution $r(V)$, then applying $f$ to produce $U = f(V)$. In general, the invertible function $f$ stretches and concentrates the density of points over the domain, thus the output samples are distributed according to a new effective distribution $q(U)$. The output density can be explicitly computed in terms of the log-det-Jacobian of $f$, $\LDJ_f$,
\begin{equation} \label{eqn:change-of-variables}
  q(U) = \frac{r(V)}{e^{\LDJ_f(V)}},
  \quad e^{\LDJ_f(V)} := \left| \det_{ij} \frac{\partial \left[f(V)\right]_i}{\partial V_j} \right|.
\end{equation}
Here, the indices $i$ and $j$ run over directions in the Lie algebra of $\GG$ translated to $f(V)$ and $V$, respectively~\cite{wijsmanChap5}.

When $f$ is parameterized\footnote{The prior $r(V)$ may also be parameterized, though parameters controlling deterministic transformations of stochastic variables, as in $f$, have been shown to be easier to optimize~\cite{Rezende2014reparam,kingma2013reparam,Titsias2014reparam}.} by a collection of model parameters $\modelparams$, the model output distribution $q(U)$ can be considered a variational ansatz for the target distribution $p(U)$. Its free parameters can be optimized to produce an approximation to the target distribution, $q(U) \approx p(U)$, by applying stochastic gradient descent to a loss function defined to be a measure of the divergence between $q(U)$ and $p(U)$. For this optimization to be viable without a large body of training data from existing samplers, we must be able to approximate the divergence and its gradients using only samples from the model and the functional form of the action. This may be achieved by employing the Kullback-Leibler (KL) divergence between the two distributions as a loss function:
\begin{equation} \label{eqn:kl-div}
  \DKL(q||p) := \int \D{U} q(U) \left[ \log{q(U)} - \log{p(U)} \right] \geq 0.
\end{equation}
For lattice theories, it is convenient to shift the KL divergence to remove the (unknown) constant $\log{Z}$, defining a modified KL divergence~\cite{LiWang2018NNRG}\footnote{This can be considered a special case of the variational lower bound~\cite{Blei_2017}.}
\begin{equation} \label{eqn:kl-div-shift}
  \DKL'(q||p) := \int \D{U} q(U) \left[ \log{q(U)} + S(U) \right] \geq -\log{Z}.
\end{equation}
The gradients and location of the minimum are unaffected by this constant shift. The KL divergence can then be stochastically estimated by drawing samples $U$ from the model and computing the sample mean of $\log{q(U)} + S(U)$, from which stochastic gradients with respect to the model parameters $\modelparams$ can be computed via backpropagation.

It is illuminating to consider the variational ansatz as defining a family of \emph{effective actions}, any of which we can directly sample, i.e., the model density can be interpreted as arising from the effective action $\Seff(U) := -\log(q(U))$.
The ability to both compute the effective action and sample from it enables producing unbiased estimates of observables under the true distribution. For example, a reweighting approach can be used~\cite{noe2019boltzmann}, in which the vacuum expectation value of an operator $\obs$ can be computed as
\begin{equation} \label{eqn:exact-reweighting}
  \begin{gathered}
  \braket{\obs} = \frac{
    \int \D{U} q(U) \left[ \obs(U) w(U) \right]
  }{
    \int \D{U} q(U) \left[ w(U) \right]
  }
  = \frac{\braket{\obs(U) w(U)}_{\Seff}}{\braket{ w(U) }_{\Seff}}, \\
  \text{where}\quad w(U) = \exp\left( -S(U) + \Seff(U) \right).
  \end{gathered}
\end{equation}
Since $\Seff$ is an approximation of the true action, the reweighting factors $w(U)$ will vary with $U$. A measure of the quality of the reweighted ensemble is the \emph{effective sample size} (ESS),
\begin{equation} \label{eqn:ess}
  \ESS := \frac{\left(\frac{1}{n} \sum_i w(U_i) \right)^2}{\frac{1}{n} \sum_i w(U_i)^2}, \quad
  U_i \sim q(U),
\end{equation}
which is normalized relative to the total number of samples $n$ such that $\ESS = 1$ for a perfect model. This reweighting approach is computationally efficient when computing observables is inexpensive relative to drawing samples from the model, because the extra cost of computing observables on samples which will be severely down-weighted is small.

When computing observables is instead expensive relative to drawing samples from the model, producing unbiased estimates of observables by resampling techniques can be more efficient than reweighting. A \emph{flow-based Markov chain} is one such approach~\cite{Albergo:2019eim,Kanwar:2020xzo}.\footnote{In some situations, either bootstrap resampling with weights (also known as Sampling Importance Resampling)~\cite{Rubin1987SIR} or rejection sampling may be useful. In the former approach, the ensemble size cannot easily be expanded, while in the latter, a multiplicative factor $M$ must be chosen such that $M q(U) \geq p(U)$ while avoiding excessive rejection; these challenges motivate the use of flow-based MCMC in this work.} In a flow-based Markov chain, samples from the model are used as proposals for each step of the chain, with a Metropolis accept/reject step to guarantee asymptotic exactness. Each proposal is independent of the previous configuration in the chain, and therefore the appropriate acceptance probability is
\begin{equation} \label{eqn:metropolis-pacc}
  p_{\text{acc}}(U \rightarrow U') = \min\left(1, \frac{p(U')}{q(U')} \frac{q(U)}{p(U)}\right).
\end{equation}
When the model closely approximates the target, $q(U) \approx p(U)$, the acceptance rate will be close to 1. Rejections duplicate the previous state of the chain, and observables only need to be computed once on each sequence of duplicated samples in the chain. Essentially, the Markov chain approach acts as an integer rounding of the reweighting factors, and thus resources are efficiently allocated towards computing observables only on sufficiently likely configurations. In the flow-based Markov chain, the analogue of the effective sample size is determined by correlations between sequential configurations; these correlations are introduced entirely through rejections, since proposals are independently drawn from the model.

The efficiency of the flow-based sampling approach hinges on implementing a general and well-parameterized function $f$, which must be invertible and for which $\LDJ_f$ must be tractable. A powerful approach to constructing such functions is through composition of simpler functions $g_i$:
\begin{equation} \label{eqn:decompose-f}
  f(V) := g_n(g_{n-1}(\dots g_1(V) \dots )).
\end{equation}
When each $g_i$ is invertible and has a tractable log-det-Jacobian, $f$ satisfies these properties as well. In the following sections, we choose the $g_i$ to be \emph{coupling layers}: functions that act elementwise on a subset of the components of the input, conditioned on the complimentary (``frozen'') subset. This structure guarantees a triangular Jacobian matrix, allowing $\LDJ_f$ to be efficiently computed from the diagonal elements of the matrix. Coupling layers generally guarantee invertibility by defining the transformation as an explicitly invertible operation on the input.
For example, a coupling layer could transform a link in a gauge configuration by left-multiplication with a group element that only depends on nearby frozen links and model parameters, $\modelparams$:
\begin{equation} \label{eqn:example-link-transform}
  \begin{gathered}
  U_\mu(x) \stackrel{\textrm{e.g.}}{\rightarrow} U'_\mu(x) = W_{\modelparams}(\text{frozen neighbors}) U_\mu(x),
  \end{gathered}
\end{equation}
where $W_{\modelparams}(\text{frozen neighbors}) \in \G$.
Regardless of the function $W_\modelparams$, this transformation is invertible: to undo it, we compute $\left[W_\modelparams(\text{frozen neighbors})\right]^{-1}$ and left-multiply.

In general, coupling layers are written in terms of functions of the frozen links and model parameters (analogous to $W_\modelparams$ in the example above), which we call \emph{context functions}. The outputs of these context functions are used to transform the input in a manifestly invertible way, but the functions themselves may be arbitrary, up to producing output in the correct domain (in our example, returning values in $\G$). These functions are therefore typically implemented in terms of feed-forward neural networks, with the model parameters $\modelparams$ specifying the neural network weights.

\subsection{Symmetries in flow models} \label{subsec:flow-based-symm}
Symmetries in a lattice gauge theory manifest as transformations of field configurations that leave the action invariant for all field configurations. We write the transformation $t$ acting on a field configuration $U$ as $t \cdot U$; a group of transformations $H$ is then a symmetry group when $S(U) = S(t \cdot U)$ for all $t \in H$ and all $U$.  Lattice actions $S(U)$ are commonly constructed to preserve discrete geometric symmetries of the Euclidean spacetime as well as internal symmetries. In particular, actions are typically invariant under the:
\begin{enumerate}
\item Discrete translational symmetry group, $T = \{ T_{\delta x, \delta y} \}$, where $\delta x, \delta y$ enumerate all possible lattice offsets;
\item Hypercubic symmetry group $R = \{ R_i \}$, where $i$ enumerates all $2^{D}(D!)$ unique combinations of rotations and reflections of the $D$-dimensional hypercube;\footnote{These operations represent the symmetry about a distinguished point on the lattice. In general, the whole geometric symmetry group is given by the combination of this group with the translational symmetry group.}
\item Gauge symmetry group, where each element $\Omega$ can be defined as a group-valued field over lattice sites, $\Omega(x) \in \G$, that transforms links of a field configuration as:
  \begin{equation}
    (\Omega \cdot U)_\mu(x) = \Omega(x) U_\mu(x) \Omega^\dagger(x+\hat{\mu}).
  \end{equation}
\end{enumerate}

Any expressive flow-based model should approximately reproduce the symmetries of the original action after optimization, even if these symmetries are not imposed in the model. Exact symmetries are recovered on average in the sampled distribution after reweighting or composing samples into a Markov chain. Nevertheless, any breaking of the symmetries in the model reflects differences between the model and target distribution, and is thus associated with sampling inefficiencies in the form of increased variance or correlations in the Markov chain. Imposing symmetries explicitly in the form of the model effectively reduces the variational parameter space to include only symmetry-respecting maps, i.e.~those that factorize the distribution. An example of such factorization is illustrated for gauge symmetry in Fig.~\ref{fig:symfact}. In many machine learning contexts, it has been found that explicitly preserving the symmetries of interest in models improves both the optimization costs and ultimate model quality~\cite{CohenWelling2016,Cohen:2019,khler2019equivariant,rezende2019equivariant,wirnsberger2020targeted, ZhangEWang2018Monge,Finzi:2020}. For example, gauge symmetry is a large symmetry group with dimension proportional to the number of lattice sites; in our study of $\U{1}$ gauge theory in Ref.~\cite{Kanwar:2020xzo}, it was shown that imposing this symmetry exactly was necessary to construct flow-based samplers of comparable or better efficiency than traditional sampling approaches. 

\begin{figure}
    \centering
   \includegraphics[width=\columnwidth]{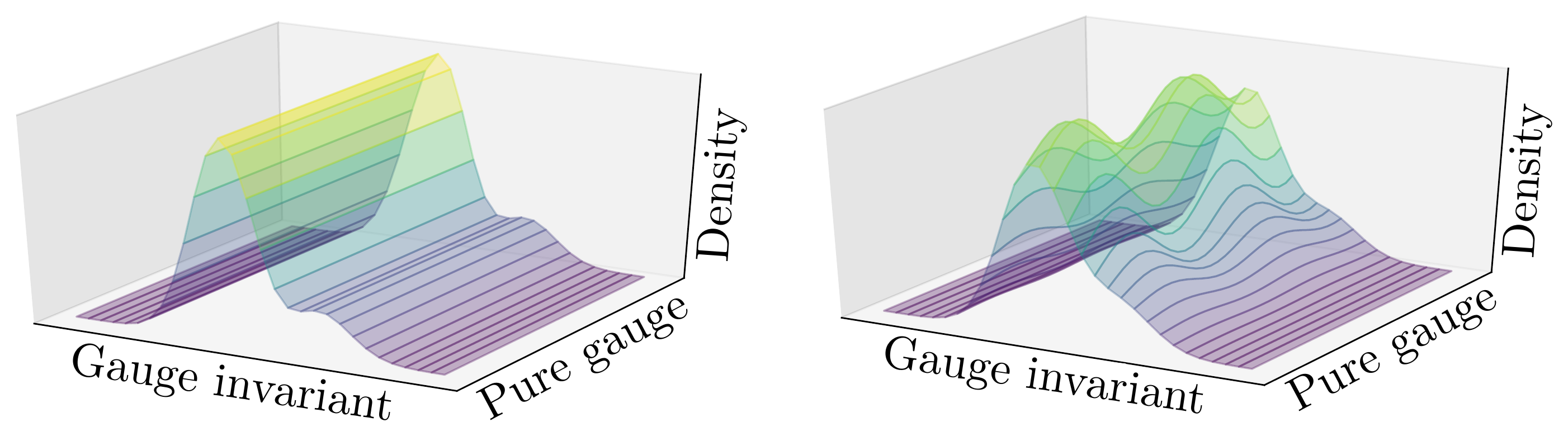}
    \caption{Left: distributions that exactly respect gauge symmetry factor over the degrees of freedom, such that they have uniform density in the pure gauge degrees of freedom and a non-trivial density only in the gauge invariant degrees of freedom. Right: arbitrary distributions on the space of gauge configurations do not factor, and uniformity in the pure gauge direction must be approximately learned by the model.}
    \label{fig:symfact}
\end{figure}

Interactions between symmetry groups are also an important consideration. For example, a simple way to achieve the factorization of the model distribution depicted in Fig.~\ref{fig:symfact} would be to employ a gauge fixing procedure that reduces configurations to gauge invariant degrees of freedom only and sample only in the remaining lower-dimensional space. This could be achieved with a maximal tree gauge fixing~\cite{Creutz1977gaugeFix,DETAR1985621}. However, gauge fixing procedures like the maximal tree procedure that explicitly factorize the degrees of freedom are not translationally invariant. On the other hand, gauge fixing procedures based on implicit differential equation constraints instead of an explicit factorization are known to preserve translational invariance in the path integral formulation~\cite{DeGrand:2006zz}, but it is unclear how to restrict flow-based models to act on configurations satisfying these constraints. Recent work in the Hamiltonian formulation has suggested ways to factor out pure gauge degrees of freedom for $\U{1}$ gauge theory, but it is not clear whether this can be extended to $\SUn$ gauge theory or the path integral formulation~\cite{Bender:2020ztu}. Here we develop an approach to simultaneously impose gauge and translational symmetries on models acting on all of the degrees of freedom of an $\SUn$ gauge field, without any preemptive factorization along the lines of gauge fixing.

To preserve a symmetry in a flow-based sampling model, it is sufficient to sample from a prior distribution that is exactly invariant under the symmetry and transform the samples using an invertible transformation that is \emph{equivariant} under the symmetry~\cite{bender2019exchangeable, rasul2019set, kohler2020equivariant}, meaning that symmetry transformations $t$ commute with application of the function,
\begin{equation} \label{eqn:equivariance}
  f(t \cdot U) = t \cdot f(U).
\end{equation}
For lattice gauge theories, a uniform prior distribution (with respect to the product Haar measure) is easily sampled and is symmetric under translations, hypercubic symmetries, and gauge symmetry. Equivariance of the map $f$ can be guaranteed by ensuring that the individual coupling layers in the decomposition of $f$ are each equivariant:
\begin{equation} \label{eqn:equivariance-decomp}
  \begin{gathered}
  g_i(t \cdot U) = t \cdot g_i(U) \\
  \implies f(t \cdot U) = g_n(g_{n-1}(\dots g_1(t \cdot U) \dots )) = t \cdot f(U).
  \end{gathered}
\end{equation}
In our approach~\cite{Kanwar:2020xzo}, coupling layers decompose the components of a field configuration by spacetime location, and therefore making coupling layers equivariant to spacetime symmetries (translational and hypercubic symmetries) and making coupling layers equivariant to internal symmetries (such as gauge symmetry) must be handled in different ways, but can be simultaneously achieved.

It has been noted that convolutional neural networks are equivariant to discrete translations, and a similar approach can extend equivariance to rotations and reflections~\cite{cohen2016group,cohen2019gauge}. For lattice gauge theory, using these equivariant networks acting on the frozen links inside each coupling layer \emph{and choosing symmetric decompositions} into frozen and updated links ensures each coupling layer is equivariant under (a large subgroup of) translations. For example, in Sec.~\ref{sec:sun-gauge-theory} we construct models for two-dimensional gauge theory using convolutional neural networks with a decomposition pattern that repeats after offsets by 4 sites in both directions on the lattice, resulting in equivariance under the translational symmetry group modulo $\mathbb{Z}_4 \times \mathbb{Z}_4$. Though the full translational symmetry group is not preserved exactly, the residual group that must be learned has a fixed size independent of the lattice volume.

Internal symmetries, on the other hand, do not mix links at different spacetime locations. The symmetry transformations acting on the frozen links already commute through the coupling layer. The updated links, however, must be transformed specifically to guarantee equivariance. Generally, this can be achieved by making the context function (i.e.~the analogue of $W_\modelparams$ acting on frozen links in Eq.~\eqref{eqn:example-link-transform}) \emph{invariant} to symmetry transformations, and defining how the function is applied to the remaining links such that the operation commutes with symmetry transformations. This must be done based on the form of the symmetry group; we review how this can be achieved for the case of gauge symmetries in the following section.

\subsection{Gauge equivariance} \label{subsec:gauge-equiv}

In Ref.~\cite{Kanwar:2020xzo}, we presented a framework for the construction of coupling layers that are equivariant under gauge symmetries. At a high level, each coupling layer is constructed to:
\begin{enumerate}
    \item Change variables to open (untraced) loops of links that start and end at a common point;
    \item Act on these loops in a way that is equivariant under matrix conjugation; we call the function acting in this way a \emph{kernel};
    \item Change variables back to links to compute the resulting action on the gauge configuration.
\end{enumerate}
Under a gauge transformation, each open loop transforms by matrix conjugation. The kernel acting on open loops is equivariant under matrix conjugation, thus the whole coupling layer is gauge equivariant. Matrix conjugation leaves the set of eigenvalues, i.e.~the \emph{spectrum}, of the open loop invariant. Arranging the coupling layer in terms of the spectra of open loops thus allows the flow to directly manipulate these physical, gauge-invariant, marginal distributions independently of the pure-gauge degrees of freedom.

In our implementation, we use $1 \times 1$ loops, or \emph{plaquettes}, as the open loops transformed by the kernel. The plaquette oriented in the $\mu\nu$ plane and located at site $x$ is defined in terms of the links by\footnote{Note that there is no trace and $P_{\mu\nu}(x)$ is matrix-valued.}
\begin{equation}
    P_{\mu\nu}(x) := U_\mu(x) U_\nu(x+\hat{\mu}) U^\dagger_\mu(x+\hat{\nu}) U^\dagger_\nu(x).
\end{equation}
A subset of plaquettes is transformed by the kernel, while the traces of unmodified plaquettes are used as gauge invariant input to the context functions in the transformation.\footnote{The use of plaquettes as the open loops and gauge invariant inputs is one of many possible choices. For either the open loops or gauge invariant inputs, plaquettes could be replaced or augmented by other choices of loops.} After the kernel acts on untraced plaquettes, $P_{\mu\nu}(x) \rightarrow P'_{\mu\nu}(x)$, we change variables back to links and implement the update on the gauge configuration as
\begin{equation} \label{eqn:transform-link-from-plaq}
    U'_\mu(x) = P'_{\mu\nu}(x) P^\dagger_{\mu\nu}(x) U_\mu(x)
\end{equation}
so that the plaquette is updated as desired,
\begin{equation}
U'_\mu(x) U_\nu(x+\hat{\mu}) U^\dagger_\mu(x+\hat{\nu}) U^\dagger_\nu(x) = P'_{\mu\nu}(x).
\end{equation}
Equivariance under matrix conjugation ensures that output plaquettes transform appropriately under the gauge symmetry,
$ (\Omega \cdot P')_{\mu\nu}(x) = \Omega(x) P'_{\mu\nu}(x) \Omega^\dagger(x) $,
and therefore the output configuration does as well:
\begin{equation}
\begin{aligned}
    (\Omega \cdot U')_\mu(x) =&\,
    \left[ \Omega(x) P'_{\mu\nu}(x) \Omega^\dagger(x) \right]
    \left[ \Omega(x) P^\dagger_{\mu\nu}(x)  \Omega^\dagger(x) \right] \\
    &\,\times
    \left[ \Omega(x) U_\mu(x) \Omega^\dagger(x+\hat{\mu}) \right] \\
    =&\, \Omega(x) U'_\mu(x) \Omega^\dagger(x+\hat{\mu}).
\end{aligned}
\end{equation}
This general construction is schematically depicted in the outer, gray sections of Fig.~\ref{fig:gauge-equiv-overview}.

\begin{figure}
    \centering
    \includegraphics[width=\linewidth]{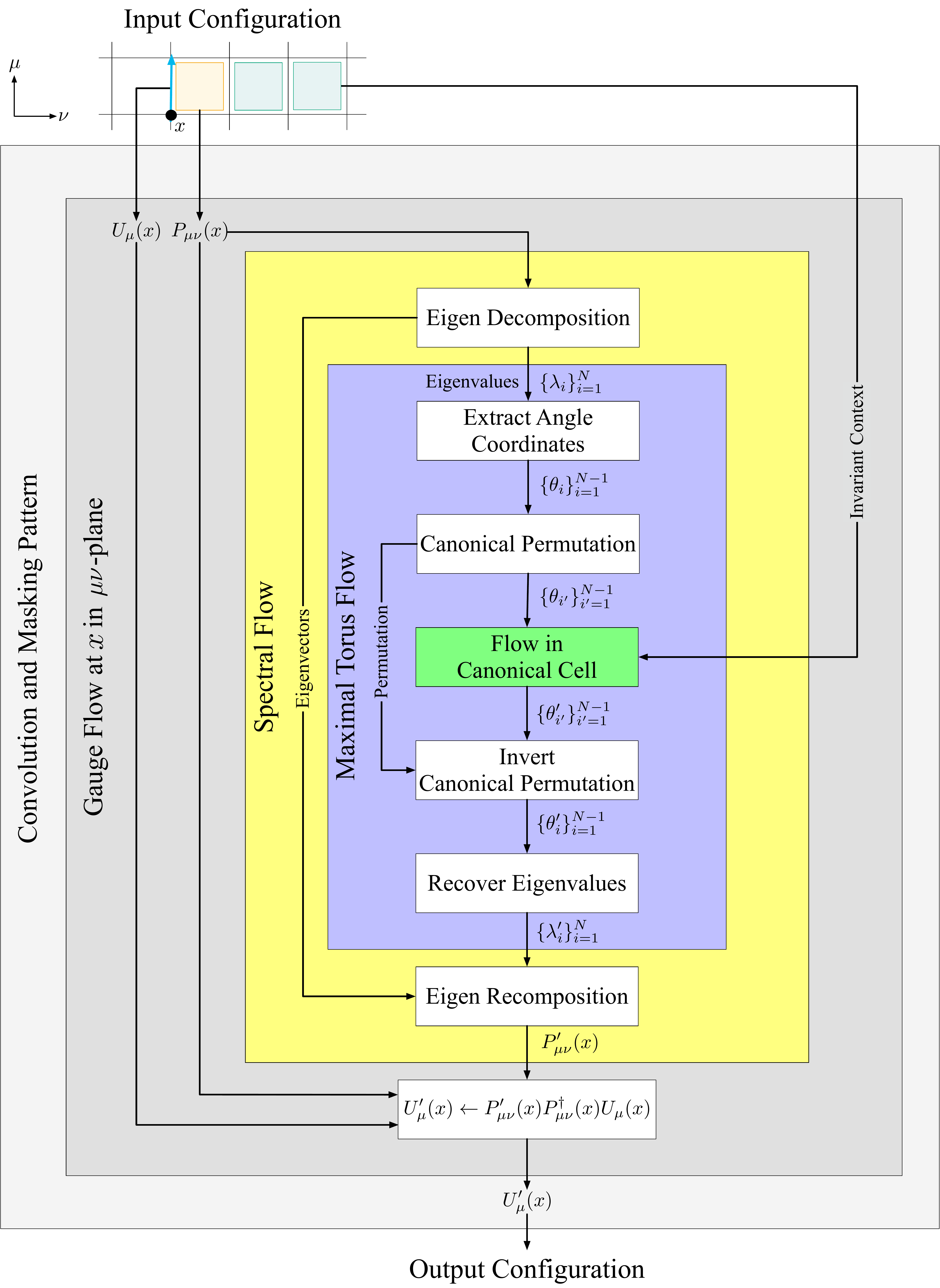}
    \caption{Decomposition of a single gauge equivariant coupling layer. Outer gray sections depict the general formulation of gauge equivariant flows detailed in Ref.~\cite{Kanwar:2020xzo}. Inner colored sections detail the kernel we construct in Sec.~\ref{sec:single-sun-flows} for a single $\SUn$ variable.}
    \label{fig:gauge-equiv-overview}
    \vspace{.2in}

    \centering
    \includegraphics[width=1.0\linewidth]{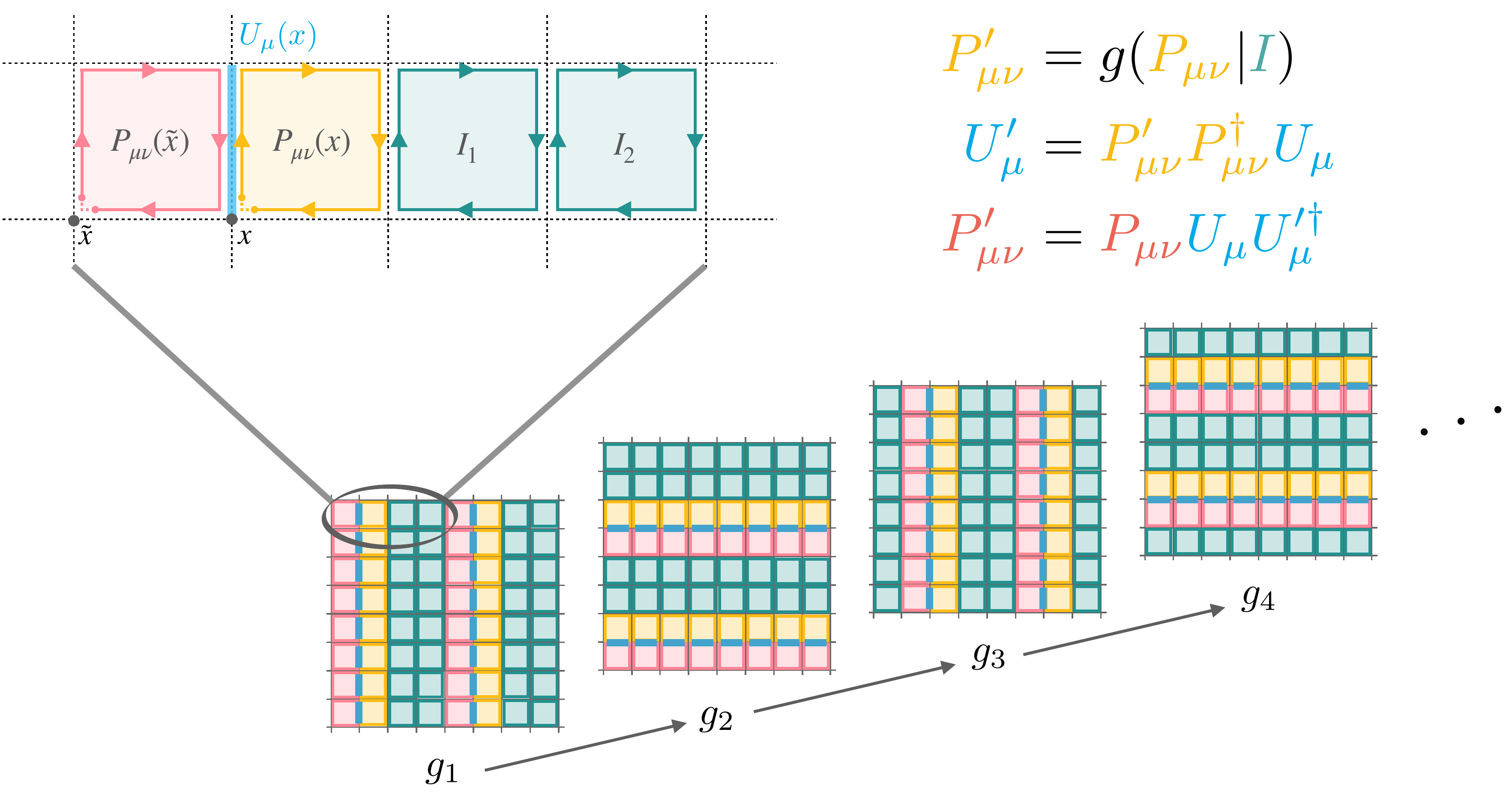}
    \caption{Our choice of plaquettes to update [$P_{\mu\nu}(x)$, yellow], gauge invariant context for that transformation [$I_1$ and $I_2$, green], the corresponding updated link [$U_\mu(x)$, blue], and the plaquettes passively modified as a result of the link update [$P_{\mu\nu}(\tilde{x})$, red] for two-dimensional gauge theory. A repeating cycle of rotations and translations are applied to the pattern for successive coupling layers; composition of 8 coupling layers is sufficient to update every link once for this pattern.} 
    \label{fig:mask_passive_active}
\end{figure}

Finally, to ensure invertibility, we require that the term $P^\dagger_{\mu\nu}(x) U_\mu(x) = U_\nu(x) U_\mu(x+\hat{\nu}) U^\dagger_\nu(x+\hat{\mu})$ in Eq.~\eqref{eqn:transform-link-from-plaq} does not contain any links that are updated as a result of other plaquettes being transformed. In our construction, we must choose the subsets of loops to transform, and the corresponding links to update, in such a way that any loop that is \emph{actively} transformed is not also modified \emph{passively} as a byproduct of another loop being transformed. There are many possible ways to choose subsets satisfying these constraints; to ensure that all links are updated, we should also choose different subsets of loops to update in each coupling layer. For example, in our application to two-dimensional gauge theory we choose to update rows or columns of plaquettes that are spaced 4 sites apart, with a repeating cycle of offsets and rotations in each successive coupling layer, as depicted in  Fig.~\ref{fig:mask_passive_active}. Note that the subsets of plaquettes that are actively and passively updated are disjoint in all coupling layers.

In Ref.~\cite{Kanwar:2020xzo}, we applied this general gauge equivariant construction to $\U{1}$ gauge theory. Our contribution in the present work is the development of transformations that are equivariant under matrix conjugation in $\SUn$ (with a straightforward adaptation to $\Un$) which can be used as kernels for gauge equivariant coupling layers in $\SUn$ or $\Un$ lattice gauge theory. This novel contribution is depicted in the inner, colored sections of Fig.~\ref{fig:gauge-equiv-overview}. We detail these transformations in the next section.

\begin{figure*}
    \centering
    \includegraphics[width=0.8\linewidth]{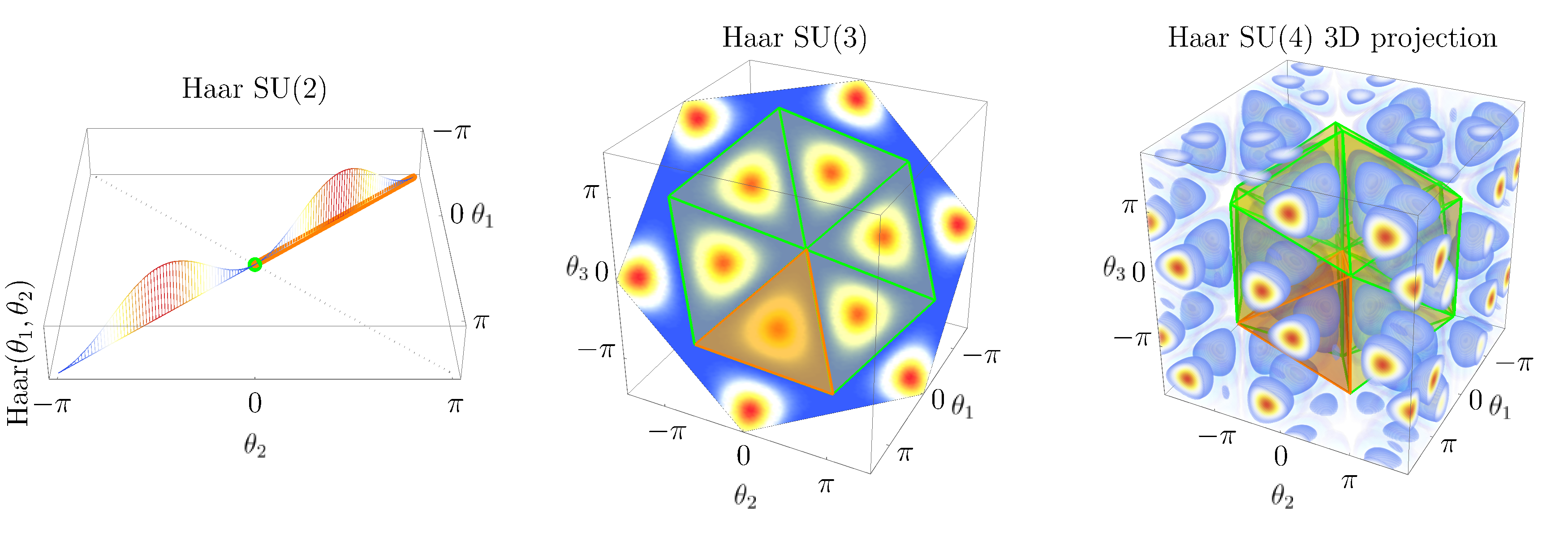}
    \caption{Illustration of the eigenvalue spaces and respective Haar measures in the angular coordinate system $\theta_k = \arg (\lambda_k) $ for $\SU{2}$ [left], $\SU{3}$ [middle], and $\SU{4}$ [right]. Eq.~\eqref{eqn:haar-density} describes how the Haar measure is included in these plots over the space of eigenvalues.
    The constraint $\det{U} = 1$ restricts the space of eigenvalues to the surface of codimension 1 defined by $\sum_k \theta_k = 0 \pmod{2\pi}$ depicted in each space. On each surface, permutation of the axes corresponds to permutation among the $N!$ cells delineated by green boundaries. A canonical cell used to construct permutation equivariant coupling layers is highlighted in orange for each surface. For $\SU{4}$, we show the surface of eigenvalues projected to an orthonormal basis in the constraint surface. For clarity in the $\SU{3}$ and $\SU{4}$ figures, we extend the range of the axes rather than showing the parts of the eigenvalue surface that would wrap around the periodic boundaries.
    }
    \label{fig:conj_equiv}
\end{figure*}

\section{Flow models for single $\SUn$ variables} \label{sec:single-sun-flows}
The key component of a gauge equivariant flow-based model is a kernel: an invertible map that acts on a single group-valued variable and is equivariant under matrix conjugation. Specifically, an invertible map $h: \G \rightarrow \G$ is a kernel if $h(X U X^{-1}) = X h(U) X^{-1}$ for all $U, X \in \G$. In constructing a gauge equivariant flow-based model, the kernel is used to transform untraced loops of links starting and ending at a common point (whose spectrum has physical, gauge-invariant meaning). Here, we specify a general method to construct such kernels and investigate application of these kernels to sampling probability densities on single $\SUn$ or $\Un$ variables (representing marginal distributions on open loops in the full gauge theory).

In the language of groups, a kernel should move density \emph{between} conjugacy classes while preserving structure \emph{within} those classes. Each conjugacy class is defined by a set $\{X U X^{-1}: X \in \G\}$, for some $U$. It is useful, however, to think of each conjugacy class in $\SUn$ or $\Un$ as a set of all matrices with some particular spectrum; for example, all matrices with eigenvalues $\{ e^{i 3 \pi / 12}, e^{i 5\pi / 12}, e^{- i 8 \pi / 12} \}$ form a conjugacy class in $\SU{3}$. Intuitively, a kernel should therefore move density between possible $N$-tuples of eigenvalues while preserving the eigenvector structure. In Appendix~\ref{sec:equiv-conj-perm} we prove that this intuition is exact: a kernel can generally be defined as an invertible map that acts on the list of eigenvalues of the input matrix, is equivariant under permutations of the eigenvalues, and leaves the eigenvectors unchanged. In our applications, we therefore structure the kernel to accept a matrix-valued input, diagonalize it to produce an (arbitrarily ordered) list of eigenvalues and eigenvectors, transform the eigenvalues in a permutation equivariant fashion, then reconstruct the matrix using the new eigenvalues. Fig.~\ref{fig:gauge-equiv-overview} depicts how this \emph{spectral flow} is applied in the context of a gauge equivariant coupling layer.

Permutation equivariance is required to ensure that the kernel acts only based on the spectrum, not the particular order of eigenvalues produced during diagonalization. Normalizing flows that are permutation equivariant have previously been investigated in the machine-learning community to learn densities over sets (such as point-clouds, objects in a 3D scene, particles in molecular dynamics, and particle tracks in collider events)~\cite{bender2019exchangeable,rasul2019set,kohler2020equivariant,guttenberg2016permutation,rahme2020permutation,ravanbakhsh2017equivariance,gordon2019permutation,maron2018invariant,segol2019universal,sannai2019universal,Komiske_2019}.
Such approaches are directly applicable to kernels for $\Un$ variables (see Appendix~\ref{sec:case_of_un}), because the eigenvalues can be transformed independently. For an $\SUn$ variable, however, the constraint $\det{U} = 1$ must additionally be satisfied, which prevents these methods from being straightforwardly applied.
As an example, Figure~\ref{fig:conj_equiv} depicts the space of eigenvalues of $\SU{2}$, $\SU{3}$, and $\SU{4}$ variables and illustrates the constrained surface of possible eigenvalues as well as the cells on this surface that are related by permutations in each case. To be equivariant, a spectral flow for $\SUn$ must transform values within each cell identically.

In this section, we describe special-case constructions of permutation equivariant transformations on the eigenvalues of an $\SU{2}$ or $\SU{3}$ variable, then generalize the approach to $\SUn$. In each case, we demonstrate the expressivity of these transformations by constructing flow-based models in terms of these transformations and training the models to learn several target families of densities that are invariant under matrix conjugation.

\subsection{Target densities}\label{subsec:target_densities}
As target distributions to test this approach, we define densities on $\SUn$ matrices that are invariant under matrix conjugation. For an $\SUn$ variable in the fundamental matrix representation, such a class of probability densities can be defined in terms of traces of powers of the variable,
\begin{equation} \label{eqn:toy-dists}
\begin{gathered}
    p_{\text{toy}}^{(i)}(U) := e^{-S_i(U)} / Z_i, \quad Z_i = \int dU e^{-S_i(U)},
\end{gathered}
\end{equation}
where
\begin{equation}\label{eq:target_action}
    S_i(U) := -\frac{\beta}{N} \Re \tr \left[ \sum_n c^{(i)}_n  U^n \right]
\end{equation}
and $\int dU$ is integration with respect to the Haar measure of the group. Any distribution in this family is manifestly invariant under matrix conjugation, and is therefore a function of the spectrum only. The coefficients $c^{(i)}$ determine the shape of the density on the group manifold, while $\beta$ determines the scale of the density.

The coefficients $c^{(i)}$ defining the target densities for this study are reported in Table~\ref{tab:toy-params}. The first set of coefficients, $c^{(0)}$, was chosen to exactly match the marginal distribution on each open plaquette in the case of two-dimensional lattice gauge theory. To further investigate densities with similar structure, two additional sets of coefficients were chosen by randomly drawing values for $c^{(i)}_1$, $c^{(i)}_2$, and $c^{(i)}_3$, and restricting to coefficients that produce a single peak in the density across all values of $\beta$. Performance on this set of coefficients is therefore representative of the ability of these flows to learn the local densities relevant to sampling for two-dimensional lattice gauge theory.

\begin{table}
    \centering
    \begin{tabular}{c @{\hskip0.3in} c @{\hskip0.2in} c @{\hskip0.2in} c}
        \toprule
        set $i$ & $c^{(i)}_1$ & $c^{(i)}_2$ & $c^{(i)}_3$ \\
        \midrule
        0 &  1 & 0 & 0 \\
        1 &  0.17 & -0.65 &  1.22  \\
        2 &  0.98 & -0.63 & -0.21 \\
        \bottomrule
    \end{tabular}
    
    \caption{Sets of coefficients $c^{(i)}_n$ used to investigate the $\SU{2}$ and $\SU{3}$ matrix conjugation equivariant flow.}
    \label{tab:toy-params}
\end{table}

To investigate the expressivity of the permutation equivariant transformations that we define, we construct flow-based models that combine a uniform prior density with one kernel defined using the equivariant transformations under study. This combination of an invariant prior distribution with application of an equivariant kernel imposes matrix conjugation symmetry on each flow-based model exactly. As a metric for the expressivity of the permutation equivariant transformations used in each kernel, we checked the ability of the flow-based models to reproduce the target densities. Measurements of the ESS and plots of the densities are used to investigate model quality.

When plotting densities in the space of eigenvalues, as in Fig.~\ref{fig:conj_equiv} above and the density plots below, we always plot with respect to the Lebesgue measure on the eigenvalues. This is a natural choice, as densities with respect to this measure are what one expects to reproduce using histograms in the space of eigenvalues. However, the full model on $\SUn$ reports densities with respect to the Haar measure.
When restricting to the space of eigenvalues, the resulting measure is absolutely continuous with respect to the Lebesgue measure with density given by the volume in $\SUn$ of conjugacy classes.
This volume is given by~\cite{Bump2004}:
\begin{equation} \label{eqn:haar-density}
    \textrm{Haar}(\lambda_1, \dots, \lambda_N) = \prod_{i < j} \left| \lambda_i - \lambda_j \right|^2.
\end{equation}
See also the Weyl integration formula and the case of $\SU3$ in~\cite{duistermaat2012lieChap3}.

\subsection{Flows on $\SU{2}$} \label{subsec:su2-flows}
The eigenvalues of an $\SU{2}$ matrix can generically be written in terms of a single angular coordinate as $\lambda_1 = e^{i \theta}$ and $\lambda_2 = e^{-i \theta}$. The permutation group $S_2$
on these eigenvalues is generated by the exchange $\lambda_1 \leftrightarrow \lambda_2$, which corresponds to $\theta \rightarrow -\theta$. We can therefore define a flow on $\theta$ which is equivariant under this transformation by separately handling the case of $\theta \in [-\pi, 0]$ and $\theta \in [0, \pi]$:
\begin{enumerate}
    \item If $\theta$ is in the first interval, negate it (otherwise, do nothing).
    \item Take the result and apply any invertible flow suitable for a variable in the finite interval $[0, \pi]$; for example, a spline flow with fixed endpoints could be applied~\cite{durkan2019neural}.
    \item If $\theta$ was negated in the first step, negate the result (otherwise, do nothing).
\end{enumerate}
In effect, this extends the action of a flow on one \emph{canonical cell}, $\theta \in [0, \pi]$, to the entire domain in a permutation equivariant fashion. The canonical cell for $\SU{2}$ is schematically depicted in the left panel of Fig.~\ref{fig:conj_equiv}. This intuition is useful to extend the method to $\SU{3}$ and generic $\SUn$ variables in the following subsections.

To investigate the efficacy of this permutation equivariant spectral flow, we constructed $\SU{2}$ flow-based models to sample from each of the families of distributions defined by Eq.~\eqref{eqn:toy-dists}, with coefficients listed in Table~\ref{tab:toy-params}, for each $\beta \in \{ 1, 5, 9 \}$. All models were constructed with a uniform prior distribution (with respect to the Haar measure of $\SU{2}$) and a single matrix conjugation equivariant coupling layer, defined using the permutation equivariant spectral flow above. The transformation on the canonical cell $[0,\pi]$ was performed with a spline flow defined using 4 knots. Each model was trained using the Adam optimizer~\cite{kingma2014adam} with gradients of the loss function in Eq.~\eqref{eqn:kl-div-shift} stochastically evaluated on batches of 1024 samples per step. Appendix~\ref{sec:backprop-diag} describes how gradients can be backpropagated through matrix diagonalization during optimization.

The densities learned by the flow-based model are compared against the target densities in Fig.~\ref{fig:su2-toy-dists}. The peaks of the distribution are very precisely reproduced by the flow-based model, and the exact symmetry between the two cells (left and right half of each plot) is apparent for both the model and target densities. Minor deviations between the model and target densities appear in the tails of the distribution, below roughly a density of $10^{-4}$. These are rarely sampled regions, thus these deviations only have a minor impact on model quality: all models reached an ESS above 97\% for all sets of coefficients, as shown in Table~\ref{tab:su2_toy_ess}.

\begin{figure}
    \centering
    \includegraphics{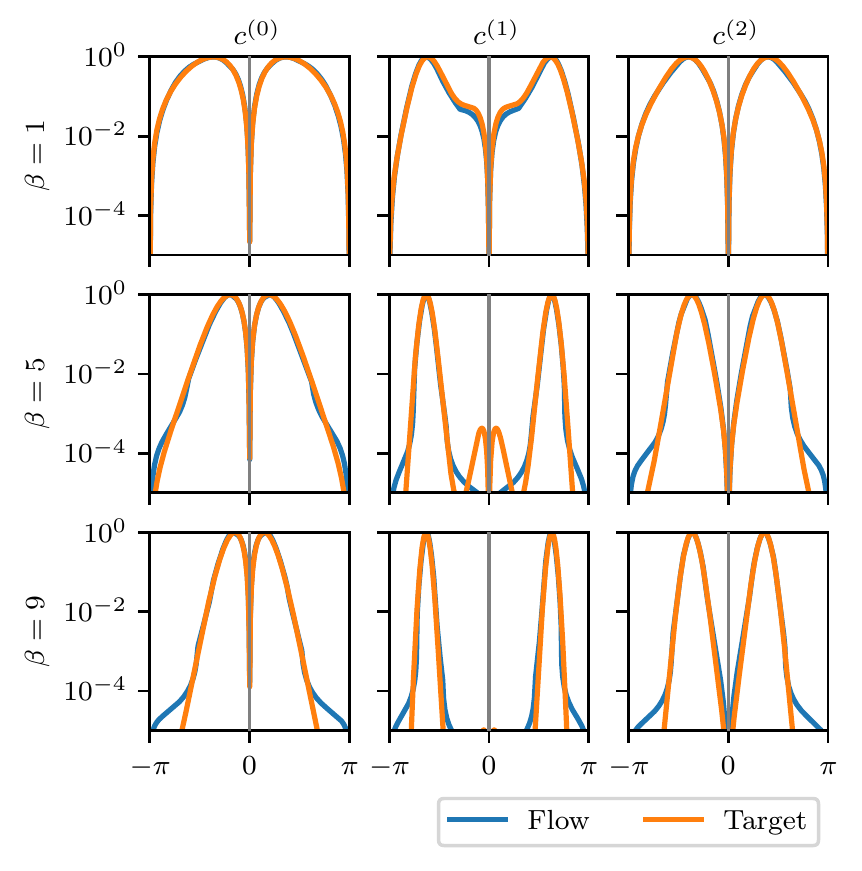}
    \caption{Densities on the angular coordinate $\theta$ describing the eigenvalues of an $\SU{2}$ variable. The mirror symmetry across $\theta = 0$ corresponds to invariance of the distribution with respect to permutation of the eigenvalues; this symmetry is exactly enforced in the flow-based distribution using a permutation equivariant coupling layer.}
    \label{fig:su2-toy-dists}
\end{figure}

\begin{table}
    \centering
    \renewcommand\arraystretch{1.2}
    \begin{tabular}{ l *{3}{
        @{\hskip0.25in}
        >{\centering\arraybackslash}p{0.16in}
        @{\hskip0.07in}
        >{\centering\arraybackslash}p{0.16in}
        @{\hskip0.07in}
        >{\centering\arraybackslash}p{0.16in}
        } }
    \toprule
    & \multicolumn{3}{c@{\hskip0.25in}}{$c^{(0)}$} & \multicolumn{3}{c@{\hskip0.25in}}{$c^{(1)}$} & \multicolumn{3}{c}{$c^{(2)}$} \\
    \midrule
    $\beta$ & 1 & 5 & 9 & 1 & 5 & 9 & 1 & 5 & 9 \\
    ESS(\%) & 100 & 100 & 100 & 98 & 98 & 97 & 100 & 99 & 100 \\
    \bottomrule
    \end{tabular}
    \caption{Final values of the ESS for each model trained for distributions on an $\SU{2}$ variable.}
    \label{tab:su2_toy_ess}
\end{table}

\subsection{Flows on $\SU{3}$}  \label{subsec:su3-flows}
The eigenvalues of an $\SU{3}$ matrix can generically be written in terms of two angular variables as $\lambda_1 = e^{i \theta_1}$, $\lambda_2 = e^{i \theta_2}$, and $\lambda_3 = e^{-i\theta_1 -i\theta_2}$. There are six cells related by the permutation group $S_3$ on these three eigenvalues, as depicted in the middle panel of Fig.~\ref{fig:conj_equiv}. We can define a permutation equivariant flow on these angular variables by extending a flow on a canonical cell to the whole space, as was done for $\SU{2}$ in the previous section:
\begin{enumerate}
    \item Enumerate all possible permutations of $[\theta_1, \theta_2, \theta_3]$, where $\theta_3 := \textrm{wrap}(-\theta_1 - \theta_2)$ is the phase of $\lambda_3$ in the interval $[-\pi,\pi]$.
    \item Choose the order $[\theta_{1'}, \theta_{2'}, \theta_{3'}]$ satisfying the canonical condition, $\textrm{iscanon}(\theta_{1'}, \theta_{2'}, \theta_{3'})$. This makes $(\theta_{1'}, \theta_{2'})$ fall in the shaded region in Fig.~\ref{fig:su3_maximal_torus_cells_-pipi}. Record the permutation required to move from the original order to the canonical order. \label{step:canon-perm-su3}
    \item Since the shaded domain in Fig.~\ref{fig:su3_maximal_torus_cells_-pipi} is split in two, replace $\theta_{1'}$ with $(\theta_{1'}-2\pi)$ if $\theta_{1'}>0$ to maintain a connected domain. Apply any invertible flow suitable for the canonical triangular domain of $\theta_{1'}$ and $\theta_{2'}$; our implementation is discussed below.
    \item Reconstruct the final angular variable $\theta'_{3'} = \textrm{wrap}(-\theta'_{1'}-\theta'_{2'})$, then apply the inverse of the permutation in step \ref{step:canon-perm-su3} to produce the final eigenvalue phases $[\theta'_1, \theta'_2, \theta'_3]$.
\end{enumerate}
For $\SU{3}$, we can define the canonical condition on eigenvalue phases in an ad hoc fashion,
\begin{equation}
    \textrm{iscanon}(\theta_1, \theta_2, \theta_3) = \begin{cases}
    \theta_3 \geq \theta_2 \geq \theta_1 & \sum_i \theta_i = 0 \\
    \theta_1 \geq \theta_3 \geq \theta_2 & \sum_i \theta_i = 2\pi \\
    \theta_2 \geq \theta_1 \geq \theta_3 & \sum_i \theta_i = -2\pi
    \end{cases}.
\end{equation}
Intuitively, this function defines a canonical ordering of the eigenvalues while smoothly accounting for the fact that they are circular variables. This intuition is made more precise in the generalization of this approach to $\SUn$ variables in the following subsection. The ad hoc shift used to move the cell to a contiguous region is also addressed when generalizing.

Mapping to and from canonical cells is one particular construction for permutation equivariant flows. Appendix~\ref{sec:su3-equiv-by-averaging} details an alternate method based on averaging over all permutations for $\SU{3}$. In that approach, equivariance is also guaranteed, but the cost scales as $N!$ making it unsuitable for large $N$.

\begin{figure}
    \centering
    \includegraphics{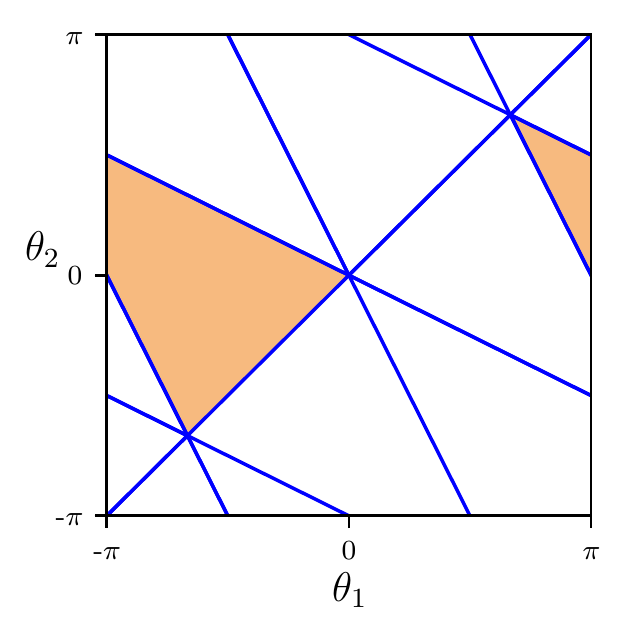}
    \caption{The cell decomposition of the maximal torus of $\SU{3}$, viewed in the $(\theta_1, \theta_2)$ coordinate system. The orange shaded cell is our choice of canonical cell.}
    \label{fig:su3_maximal_torus_cells_-pipi}
\end{figure}

We investigated the efficacy of this permutation equivariant spectral flow by constructing $\SU{3}$ flow-based models to sample from the families of distributions defined by Eq.~\eqref{eqn:toy-dists}, with coefficients listed in Table~\ref{tab:toy-params}, for each $\beta \in \{1, 5, 9\}$. All models were constructed with a uniform prior distribution (with respect to the Haar measure of $\SU{3}$) and a single matrix conjugation equivariant coupling layer, defined using the spectral flow above. The transformation on the triangular canonical cell was performed using two spline flows with 4 knots each, independently acting on the height and width coordinates. Each model was trained using the Adam optimizer with gradients of the loss function in Eq.~\eqref{eqn:kl-div-shift} stochastically evaluated on batches of 1024 samples per step.

Fig.~\ref{fig:su3_poly_dists} compares the distributions learned by the flow-based models to the target distributions when $\beta = 9$. The structure of the peaks of the distribution are reproduced accurately, and the exact six-fold symmetry between the cells is apparent in both the model and target densities. Minor deviations between the model and target densities appear in the tails of the distribution, below roughly a density of $10^{-3}$. As with the $\SU{2}$ models, these deviations are in rarely sampled regions and therefore only have a minor impact on model quality. Quantitatively, our flow-based models achieved ESSs greater than $73\%$ on all distributions, with the full set of final ESS values reported in Table~\ref{tab:su3_toy_ess}. The leftmost distribution is the marginal distribution on plaquettes for two-dimensional $\SU{3}$ gauge theory; the high value of the ESS for this distribution indicates that this spectral flow is well-suited to learn such distributions in the lattice gauge theory.

\begin{figure}
    \centering
    \includegraphics{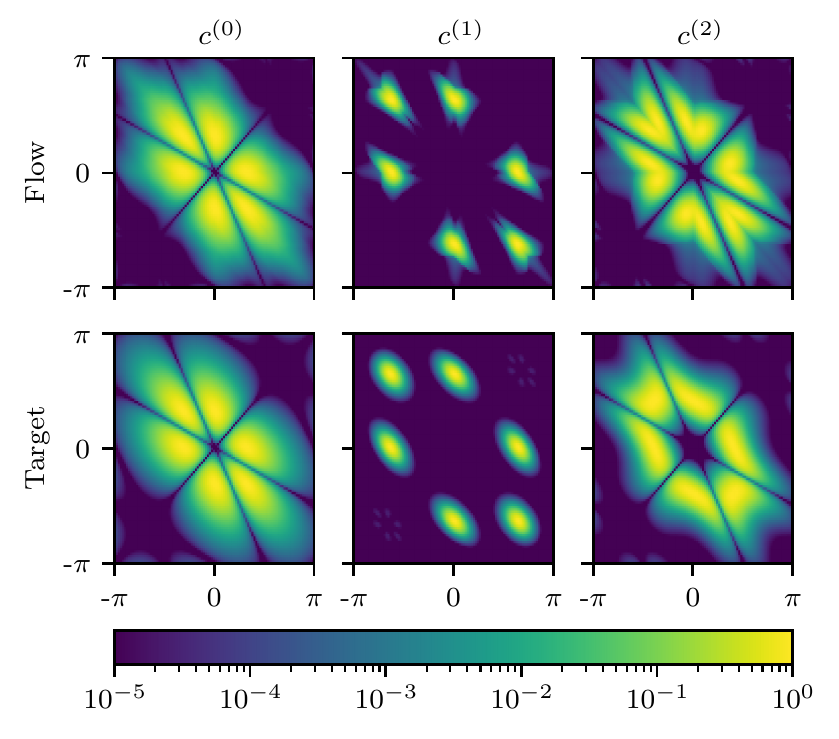}
    \caption{Densities on the angular coordinates $\theta_1$ and $\theta_2$ defining the eigenvalues of an $\SU{3}$ variable. The densities learned by the flow-based models are compared to the target densities for three distributions, each with $\beta=9$. The six-fold symmetry in each density is due to permutation invariance; this symmetry is exactly enforced in the flow-based distributions by using permutation equivariant coupling layers.}
    \label{fig:su3_poly_dists}
\end{figure}

\begin{table}
    \centering
    \renewcommand\arraystretch{1.2}
    \begin{tabular}{ l *{3}{@{\hskip0.25in}c@{\hskip0.07in}c@{\hskip0.07in}c} }
    \toprule
    & \multicolumn{3}{c@{\hskip0.25in}}{$c^{(0)}$} & \multicolumn{3}{c@{\hskip0.25in}}{$c^{(1)}$} & \multicolumn{3}{c}{$c^{(2)}$} \\
    \midrule
    $\beta$ & $1$ & $5$ & $9$ & $1$ & $5$ & $9$ & $1$ & $5$ & $9$ \\
    ESS(\%) & 99 & 98 & 99 & 97 & 80 & 82 & 99 & 91 & 73 \\
    \bottomrule
    \end{tabular}
    \caption{Final values of the ESS for each model trained for distributions on an $\SU{3}$ variable.}
    \label{tab:su3_toy_ess}
\end{table}

\subsection{Flows on $\SUn$}  \label{subsec:sun-flows}

To apply the method to $\SUn$ variables for any $N$, we develop a general version of three of the steps used above:
\begin{enumerate}
    \item Computing the vertices bounding a canonical cell;
    \item Mapping eigenvalues into that canonical cell;
    \item Applying spline transformations within that cell.
\end{enumerate}

We define cells in $\SUn$ as subsets of the maximal torus $T$, the subgroup of diagonal matrices of $\SUn$, as follows. An element of $T$ is called regular if it has $N$ distinct eigenvalues~\cite{duistermaat2012lieChap3}. The set of regular matrices in $T$ is an open set with $N!$ connected components; the closure of each component is a \emph{cell}.

To construct a general spectral flow for $\SUn$, we first choose a particular cell, which we call the \textit{canonical cell}. It is helpful to define the canonical cell in the Lie algebra $\torusalg$ of the maximal torus, rather than on the maximal torus directly. The Lie algebra $\torusalg$ is the $(N-1)$-hyperplane $\sum_{k=1}^N \theta_k=0$ in $\mathbb{R}^N$ and is related to the maximal torus by the exponential map $\textrm{exp}(\theta_1,\ldots,\theta_N)=\textrm{Diag}(e^{2\pi i\theta_1},\ldots,e^{2\pi i\theta_N})$. 
In this space, cells are $(N-1)$-simplexes enclosed by $(N-2)$-hyperplanes, each defined by a pair of eigenvalues becoming degenerate, i.e.~$\theta_j = \theta_k \pmod{2\pi}$ for some $j$ and $k$.

The $N$ vertices of any of these simplexes are mapped by the exponential map to the $N$ elements of the center of $\SUn$ (which are also elements of $T$).
We define one such simplex $\Psi$ by defining the bounding vertices $y_1,\ldots,y_N$ inside $\torusalg$,
\begin{equation}
    [y_k]_j := 2\pi \left( \frac{k}{N} - \delta_{k\geq j} \right). \label{eq:canon-cell-vert}
\end{equation}
A proof that $\exp(\Psi)$ is a cell and a derivation of this formula is given in Appendix~\ref{sec:canon-alg}. Thus we choose $\textrm{exp}(\Psi)$ as our canonical cell.

There are $N!$ ways of reordering the eigenvalues of a regular point $x=\textrm{Diag}(\lambda_1,\ldots,\lambda_N)$ in $T$, and exactly one of those falls in the canonical cell. It is intractable for large $N$ to find the element that falls in the canonical cell by checking all permutations, as we did for $\SU{3}$. Instead, we explain in Algorithm~\ref{alg:canoncell} an approach to find the pre-image in $\Psi$ of this canonical element based on sorting.
\begin{algorithm}[H]
\caption{Map into simplex $\Psi$ \label{alg:canoncell}}
canon($\lambda_1,\ldots,\lambda_N$):
\begin{enumerate}
\item extract angles in range $[0,2\pi)$, $\theta_k=\arg(\lambda_k) \mod 2\pi$
\item set $S = \frac{1}{2\pi} \sum_k \theta_k$; it is an integer because $\det{U} = 1$
\item sort the angles in ascending order $\theta^{\textrm{sort}} = \textrm{sort}(\theta)$ \label{step:sort1}
\item snap the angles to the hyperplane $\torusalg$ by\\ $\theta^{\textrm{snap}}=(\theta^{\textrm{sort}}_1,\ldots,\theta^{\textrm{sort}}_{N-S+1}-2\pi,\ldots,\theta^{\textrm{sort}}_N-2\pi)$ \label{step:sub}
\item set $\theta^{\textrm{canon}} = \textrm{sort}(\theta^{\textrm{snap}})$ \label{step:sort2}
\item return $\theta^{\textrm{canon}}$ and the combined permutation that was used to sort in step \ref{step:sort1} and \ref{step:sort2}.
\end{enumerate}
\end{algorithm}
The output of Algorithm~\ref{alg:canoncell} is a point in $\Psi$ (see Section~\ref{sec:proof_algo1}) and a permutation. To invert the map into the canonical cell after we apply a flow, we permute the flowed values $\theta'_k$ using the inverse of the returned permutation, then map them to the torus using the exponential map. Appendix~\ref{sec:proof_algo1} proves that this algorithm maps into the correct simplex. We can then show that applying the algorithm to any point in some cell returns the same output permutation by checking that the permutation does not change along any connected path within the cell. No two eigenvalues become degenerate along such a path, therefore the order of the eigenvalue phases only changes when some $\theta_k$ crosses the boundary between $0$ and $2\pi$. For example, when $\theta_k$ crosses from 0 to 2$\pi$, it will become the largest angle (instead of the smallest) and $S$ increments by 1; thus the value of $\theta_k$ after snapping changes smoothly due to the additional $2\pi$ subtracted in step \ref{step:sub}, all other angles are unaffected, and the final permutation is unchanged. A similar argument can be made when angles cross from $2\pi$ to 0.

\begin{figure}
    \centering
    \includegraphics[width=\linewidth]{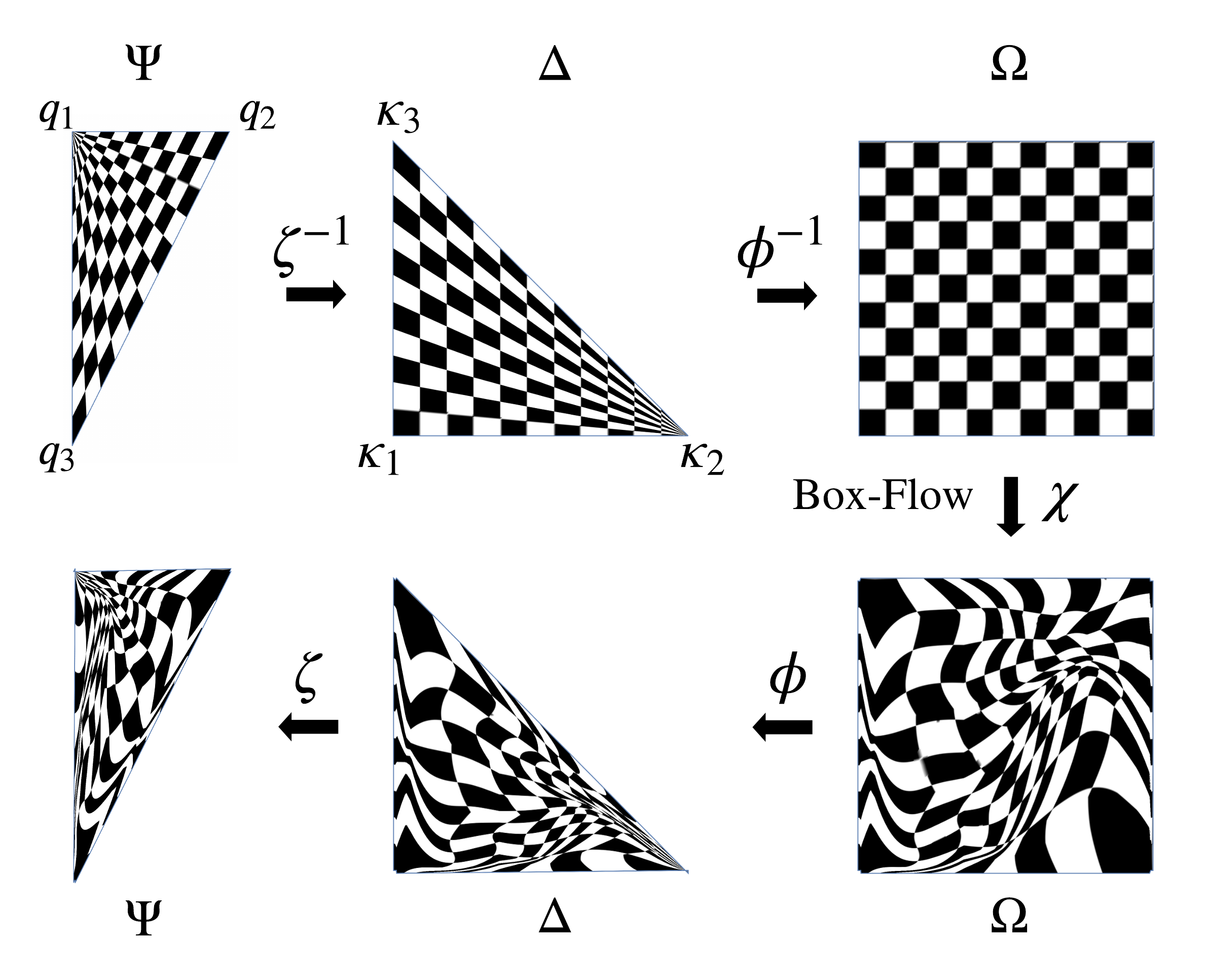}
    \caption{Illustration of the steps we use to apply a flow to an $(N-1)$-simplex, shown for $N=3$ as an example. Starting from an initial density on the simplex $\Psi$, we map it to an axis-aligned simplex $\Delta$ then to an open box $\Omega$. We apply a parametric boundary preserving flow $\boxflow$ to the box and finally invert the chain back to the original coordinate system.
    }
    \label{fig:simplex_flow}
\end{figure}

Finally, we describe the implementation of a flow on $\Psi$. To be invertible, the flow must preserve the boundaries of $\Psi$. We implement such a flow by first using a coordinate transformation to map $\Psi$ to an open box $\Omega = (0,1)^{N-1}$. On this box, an arbitrary boundary-preserving flow $\boxflow: \Omega \rightarrow \Omega$ can easily be applied (e.g.~by using transformations suitable for a finite interval along each axis). Finally, the coordinate transformation can be undone to map back to $\Psi$. It is helpful to further introduce an intermediate $(N-1)$-simplex $\Delta$, which is a right-angled simplex with equal leg lengths. Its vertices are $\{\kappa_1, \ldots, \kappa_N\}$, where $\kappa_1$ is the origin and $[\kappa_{k}]_{j} = \delta_{(k-1) j}~\forall k \in \{2,\ldots,N\}$. The map $\phi: \Omega \rightarrow \Delta$ maps the box $\Omega$ to the simplex $\Delta$ by collapsing one end of the box in each direction,
\begin{align}\label{eq:phi}
\phi_i (\alpha) &= \left\{ \begin{array}{cc}
  \alpha_1 & i = 1\\
  \alpha_i \prod_{j = 1}^{j < i} (1 - \alpha_j) & i > 1,
\end{array} \right.
\end{align}
where $\alpha \in \Omega$. The map $\zeta: \Delta \rightarrow \Psi$ then sends the intermediate right-angled simplex to the canonical simplex by
\begin{align}\label{eq:zeta}
\zeta(\rho) &= y_1 + \rho M,
\end{align}
where $\rho \in \Delta$ and $M$ is the $(N-1)\times N$ matrix defined by $M_{i j} = [y_{i+1}]_j - [y_1]_j$. Both maps are invertible. The inverse map $\phi^{-1}: \Delta \rightarrow \Omega$ is given by
\begin{align}
\phi_i^{-1}(\rho) &= \frac{\rho}{1 - \sum_{j=1}^{i-1} \rho_j},
\end{align}
for $\rho \in \Delta$, while $\zeta^{-1}: \Psi \rightarrow \Delta$ is given by
\begin{align}
\zeta^{-1}(x) &= (x - y_1)M^T (M M^T)^{- 1}.
\end{align}
The entire chain of coordinate transformations, flow, and inverse coordinate transformations is depicted in Fig.~\ref{fig:simplex_flow}.

The Jacobian of the entire flow can be computed by composing the Jacobian factors from each transformation in the chain. While the Jacobian factors acquired from the coordinate transformations are fixed, the flow acting on $\Omega$ is parameterized by, and the resulting density depends on, the action of this inner flow. For example, the inner flow could be a spline flow~\cite{durkan2019neural} constructed to transform each coordinate of $\Omega$ as a function of the model parameters and possibly the other coordinates of $\Omega$. It is this inner flow that must be trained in each coupling layer to reproduce the final density on $\SUn$. A complete listing of the algorithm to apply the matrix conjugation equivariant kernel defined by the above spectral flow is given in Appendix~\ref{sec:canon-alg}.

We implemented this general approach to matrix conjugation equivariant flows on $\SUn$ variables for a range of $N$. For $N \leq 9$, we trained these flows to reproduce target densities defined by Eq.~\eqref{eqn:toy-dists}, with coefficients listed in Table~\ref{tab:toy-params}, and $\beta = 9$. An ESS of greater than $5\%$ was achieved on all target densities, with $c^{(0)}$ performing significantly better with greater than $90\%$ ESS across all densities. Fig.~\ref{fig:toy_sun_densities} compares the flow-based densities to the target densities for $N = 9$. Worse performance on $c^{(1)}$ and $c^{(2)}$ is reflective of their multimodal nature for some values of $\beta / N$. To investigate performance at large $N$, we trained flows to reproduce the $c^{(0)}$ density for $10 \leq N \leq 100$, and found ESSs greater than $90\%$ for all models. All model distributions were confirmed to have exact permutation invariance.

\begin{figure}
    \centering
    \includegraphics{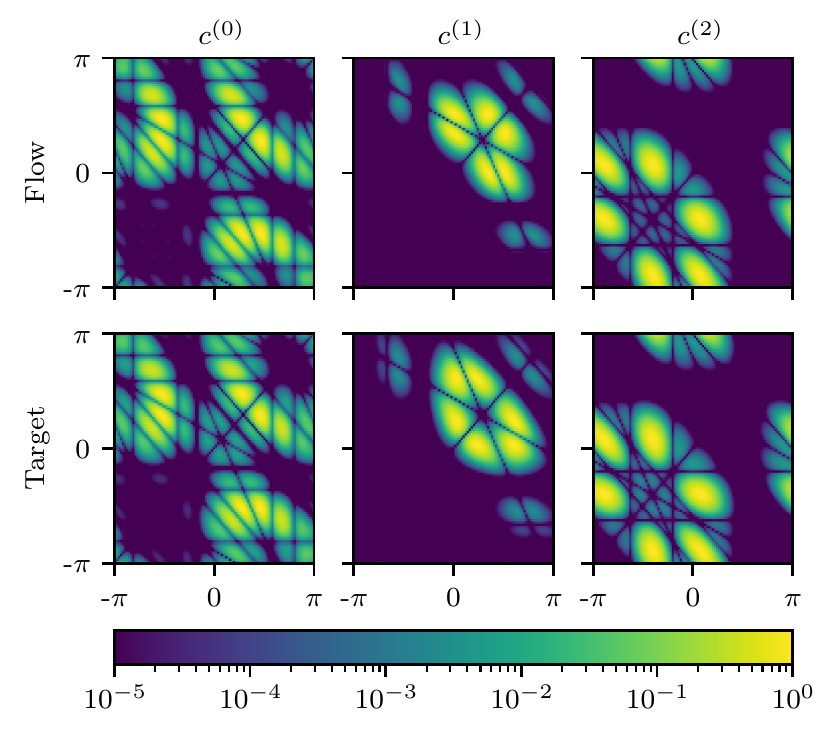}
    \caption{Densities on a two-dimensional slice through the space of $\SU{9}$ eigenvalues defined by varying $\theta_1$ and $\theta_2$, keeping $\theta_3, \dots, \theta_8$ fixed to random values, and assigning $\theta_9 = \textrm{wrap}(-\sum_{k=1}^{8} \theta_k)$. The densities learned by the flow-based models are compared to the target densities for three distributions, each with $\beta = 9$. Horizontal, vertical, and diagonal lines of zero density correspond to locations where the chosen slice crosses through walls of the cells (on which the Haar measure forces the density to zero). Due to exact permutation invariance of the flow-based distribution, these lines are exactly reproduced.}
    \label{fig:toy_sun_densities}
\end{figure}

\section{Application to $\SU{2}$ and $\SU{3}$ lattice gauge theory in 2D} \label{sec:sun-gauge-theory}
With an invertible kernel that is equivariant under matrix conjugation, the methods presented  in Ref.~\cite{Kanwar:2020xzo} immediately allow construction of gauge equivariant coupling layers for $\SUn$ lattice gauge theory. To study the efficacy of such coupling layers for this application, we trained flow-based models to sample from distributions relevant for $1+1$-dimensional gauge theory. Specifically, we considered the distribution defined by the imaginary-time path integral in Eq.~\eqref{eq:vacevalO} with the action given by the Wilson discretization of the continuum gauge action,
\begin{equation}
\begin{gathered}
    S(U) := -\frac{\beta}{N} \sum_{x} \Re \tr \left[ P_{01}(x) \right].
\end{gathered}
\end{equation}
We investigated both $\SU{2}$ and $\SU{3}$ models tuned to approximately equivalent 't Hooft couplings $\lambda = 2 N^2 / \beta$ on $16 \times 16$ periodic lattices. The parameters considered for both $\SU{2}$ and $\SU{3}$ gauge theory are listed in Table~\ref{tab:ens_params}.

\begin{table}
    \centering
    \begin{tabular}{c @{\hskip0.25in} c @{\hskip0.15in} c @{\hskip0.1in} c @{\hskip0.15in} c}
        \toprule
         $\SUn$ & $L$ & $\beta$ & $\lambda = 2N^2 / \beta$ & $n_{\text{dof}}$ \\
         \midrule
         $\SU{2}$ & $16$ & $\{1.8, 2.2, 2.7\}$ & $\{ 4.4, 3.6, 3.0 \}$ & 1536 \\
         $\SU{3}$ & $16$ & $\{4.0, 5.0, 6.0\}$ & $\{ 4.5, 3.6, 3.0 \}$ & 4096 \\
         \bottomrule
    \end{tabular}
    \caption{Choices of parameters on which we investigated the performance of our flow-based sampler. We selected three values of $\beta$ for both $\SU{2}$ and $\SU{3}$ gauge theory, corresponding to approximately equivalent 't Hooft couplings $\lambda$. $n_{\text{dof}} = D L^2 (N^2 - 1)$ indicates the dimensionality of the gauge configuration manifold in each case.}
    \label{tab:ens_params}
\end{table}

In the following subsections, we describe the architecture and training of our flow-based models, confirm the exactness of results using our sampler, and demonstrate that all symmetries are either exactly built into the model or are approximately learned by the model.

\subsection{Model architecture and training} \label{subsec:model-arch-and-training}
In all cases, we constructed flow-based models using a uniform prior distribution $r(V)$, for which configurations in the matrix representation are easily sampled.\footnote{To sample the prior distribution, the method presented in Ref.~\cite{Mezzadri2006} can be used for $\Un$ and can also be modified to fix the determinant to 1 for $\SUn$.}
The invertible function $f$ acting on samples from the prior was composed of $48$ coupling layers $g_1, \dots, g_{48}$. We constructed each coupling layer using the gauge equivariant architecture presented in Sec.~\ref{subsec:gauge-equiv}. Coupling layers specifically acted on plaquettes as the choice of open loops, transforming rows or columns of plaquettes spaced four sites apart on the lattice in each coupling layer, as denoted by $P_{\mu\nu}(x)$ (yellow) in Fig.~\ref{fig:mask_passive_active}; plaquettes that were unaffected by the transformation were used as the gauge invariant inputs to the inner spectral flow, as denoted by $I_1$ and $I_2$ (green) in Fig.~\ref{fig:mask_passive_active}. Coupling layers used an alternating sequence of rotations and a sweep over all possible translations of the transformation pattern to ensure that every link was updated after every eight layers.

The updating pattern that we define here is just one of many possible choices. One could vary the open loops that are updated, change how the links are updated as a function of the open loops, choose a different pattern of translations and rotations between coupling layers, or alter which closed loops are passed as gauge invariant inputs to context functions. The choices made here were sufficient to learn distributions in two-dimensional gauge theory, but in generalizing beyond this proof-of-principle study, in particular to higher spacetime dimensions, these choices must be studied more carefully.

For $\SU{2}$ gauge theory, we implemented the spectral flow itself in a permutation equivariant fashion as described in Sec.~\ref{subsec:su2-flows}. The flow acting on the interval $\theta \in [0, \pi]$ was a spline flow consisting of 4 knots, with the positions of the knots in $[0,\pi]$ computed as a function of the gauge invariant neighboring plaquettes $I_1, I_2, \dots$ using convolutional neural networks with 32 channels in each of 2 hidden layers. The model parameters defining the variational ansatz distribution consisted of the weights in these convolutional neural networks across all coupling layers.

For $\SU{3}$ gauge theory, we implemented the spectral flow as described in Sec.~\ref{subsec:su3-flows}. The inner flow acted on eigenvalues in the canonical triangular cell by changing coordinates to an open box and applying a spline flow in that space, as discussed in Sec.~\ref{subsec:sun-flows}. The spline flow acted on the open box in two steps, transforming the horizontal coordinate first, then the vertical coordinate conditioned on the new horizontal coordinate. The 16 knots of the splines were computed as a function of the gauge invariant neighboring plaquettes, and in the second step as a function of the horizontal coordinate as well. These two functions in each coupling layer were implemented using convolutional neural networks with 32 channels in each of 2 hidden layers. The model parameters defining the variational ansatz distribution thus consisted of the weights in this pair of convolutional neural networks across all coupling layers.

\begin{figure*}
    \centering
    \includegraphics{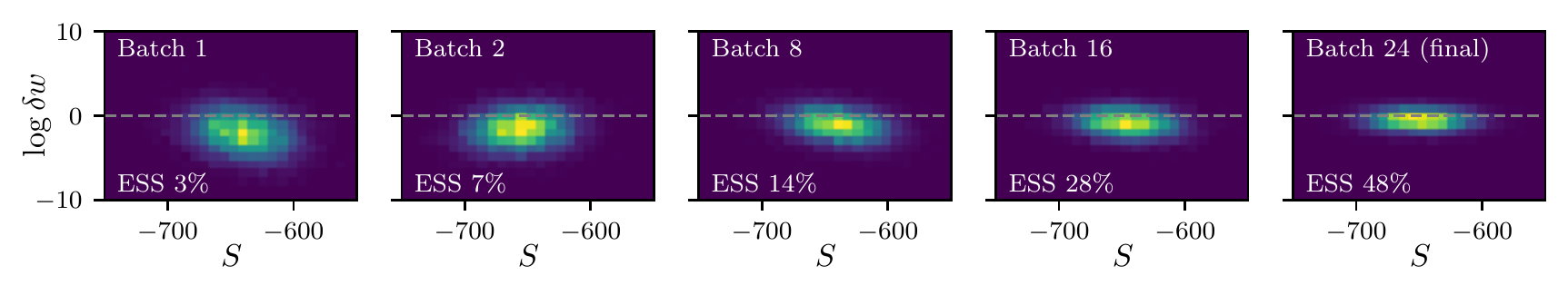}
    \caption{Normalized reweighting factors $\log{\delta w} = -S(U) + \Seff(U) - \log{Z}$ vs action $S$ per configuration across $10,000$ model proposals for $SU(3)$ gauge theory with $\beta = 6.0$ . Reweighting factors are plotted at various points throughout training, reported in terms of the number of batches of size $3072$ that have been used at that point in training.}
    \label{fig:logp_vs_logq}
\end{figure*}

In both cases, the model parameters were optimized using the Adam optimizer. Each optimization step consisted of sampling a batch of size 3072, estimating the modified KL divergence in Eq.~\eqref{eqn:kl-div-shift}, then using the optimizer to update the parameters. During training, we monitored the ESS on each batch to assess model quality. Fig.~\ref{fig:logp_vs_logq} shows how ESS and the spread of reweighting factors evolve over the course of training on a representative model. The final values of ESS for each model are reported in Table~\ref{tab:final_ess}.

\begin{table}
    \centering
    \renewcommand\arraystretch{1.2}
    {\setlength\tabcolsep{5pt}
    \begin{tabular}{l c c c c c c}
    \toprule
    & \multicolumn{3}{c}{$\SU{2}$} & \multicolumn{3}{c}{$\SU{3}$} \\
    \midrule
    $\beta$ & $1.8$ & $2.2$ & $2.7$ & $4.0$ & $5.0$ & $6.0$ \\
    ESS(\%) & 91 & 80 & 56 & 88 & 75 & 48 \\
    \bottomrule
    \end{tabular}}
    \caption{Final values of the ESS for each model trained for $\SU{2}$ and $\SU{3}$ gauge theory.}
    \label{tab:final_ess}
\end{table}

For this proof-of-principle study, we did not perform an extensive search over training hyperparameters. When scaling the method, we expect careful tuning of these hyperparameters and the model architecture can improve the model quality and allow more efficient training. Automatic tuning of hyperparameters, in particular, have been shown to significantly reduce model training costs in other machine learning applications \cite{snoek2012bayesian,BalaprakashDeepHyper,BalaprakashDeepHyper2}.

In general, models defined in terms of convolutional neural networks acting on invariant quantities in a localized region capture the local correlation structure defining the distribution. This local structure is independent of volume as long as finite volume effects are not too large. Thus models can largely be trained on much smaller volumes than the target volume, with very few training steps required at the final volume to correct for any finite volume effects learned by the model.

\begin{figure}
    \centering
    \includegraphics[width=\linewidth]{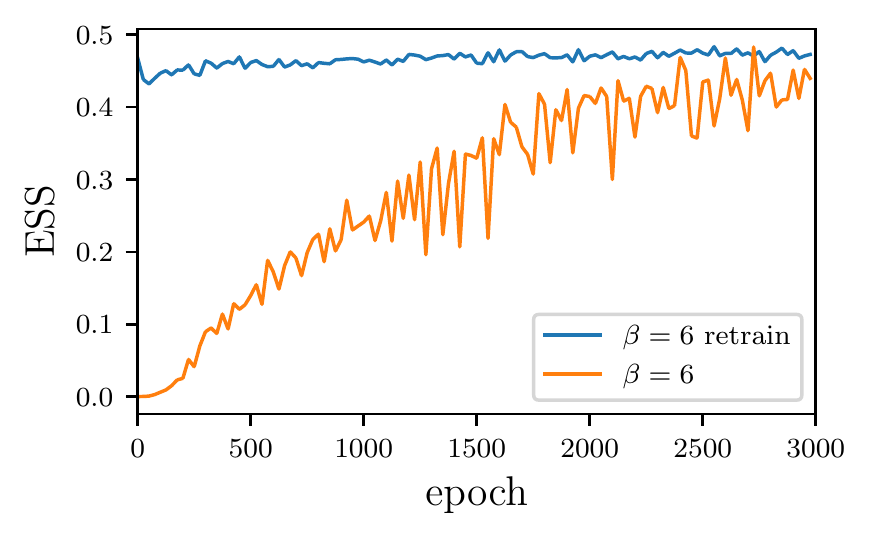}
    \caption{Comparison of training dynamics for a model for $\SU{3}$ gauge theory on a $16\times16$ lattice, when initialized with weights from a model trained on an $8\times8$ geometry, versus the dynamics for an identical model trained from a random initialization. Results are shown for the $\beta=6$ target in $\SU{3}$ gauge theory. The model transferred from the $8\times8$ geometry almost immediately converges to a plateau in model quality, while the model trained from a random initialization requires many training steps to converge to similar quality, despite adjusting the optimization hyperparameters to improve the rate of optimization from a random initialization.}
    \label{fig:train_dynamics}
\end{figure}

\begin{figure}
    \centering
    \includegraphics{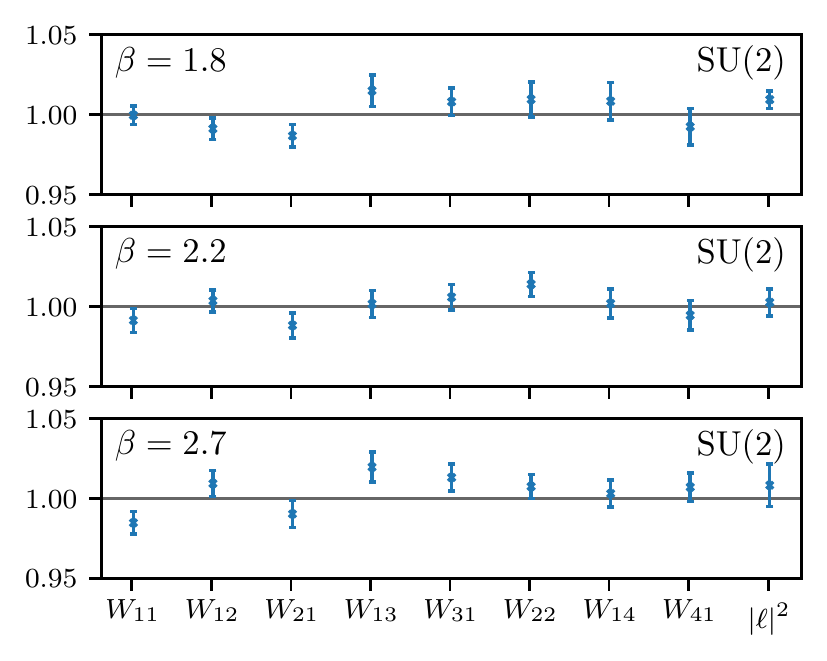} \\
    \includegraphics{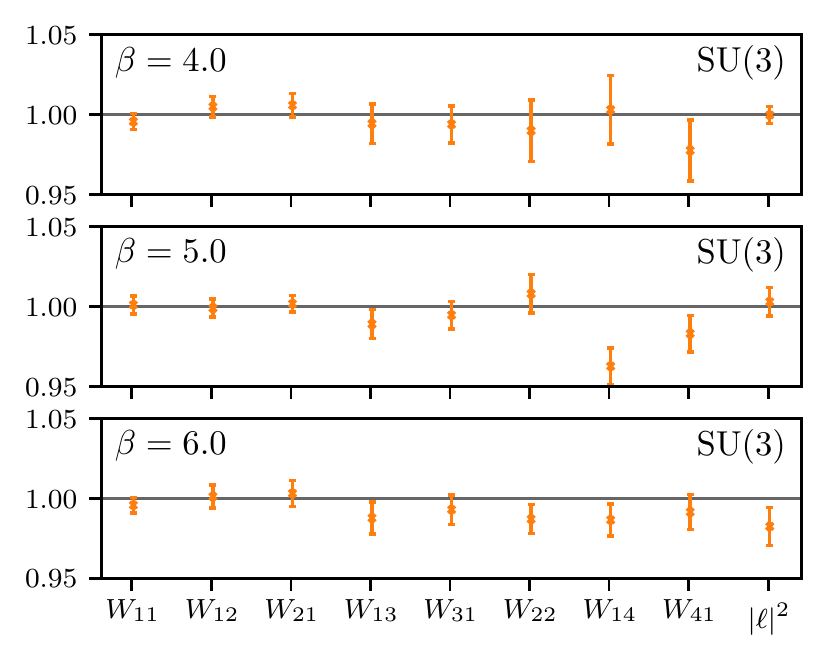}
    \caption{Selection of observables, relative to the true values, computed using the flow-based $\SU{2}$ [top] and $\SU{3}$ [bottom] gauge theory ensembles. Observables were measured on configurations separated by a number of steps equal to the Markov chain autocorrelation time, as determined by the self-consistent estimator presented in Ref.~\cite{Wolff:2003sm}. The autocorrelation time ranged from 1 to 4 for all observables. Per observable, a total number of samples ranging from 20 to 15000 was chosen to give percent-level errors.}
    \label{fig:obs_su2_su3_all}
\end{figure}

The two-dimensional gauge theories used to investigate this model consist entirely of ultralocal dynamics, with any finite volume corrections exponentially small in the number of lattice sites~\cite{sigdel2016}. In our study, we were thus able to train nearly optimal models on much smaller volumes, which enabled significantly more efficient training. For example, Fig.~\ref{fig:train_dynamics} shows that transferring a model that has already learned to capture the local structure of the $\beta=6$, $\SU{3}$ gauge theory on an $8 \times 8$ lattice almost immediately results in an optimized model for the target $16 \times 16$ lattice geometry, whereas it takes many training steps at the large volume to reach similar model quality when beginning training from a randomly initialized model. In any theory with a mass gap $M$, we expect that finite volume effects will be exponentially small in $ML$ when the side length $L$ of the lattice is large enough. Initially training at the smallest value of $L$ in this regime thus provides an efficient approach to training models with larger $L$ since the corrections that must be learned are exponentially small. These gains will be even more significant in higher spacetime dimensions, where the number of lattice sites scales with a larger power of the lattice side length $L$.

\subsection{Observables}
For each model, we constructed a flow-based Markov chain using independent proposals from the model with a Metropolis accept/reject step, as described in Sec.~\ref{subsec:flow-based-review}. Composing proposals into a Markov chain in this way ensures exactness in the limit of infinite statistics.

At finite statistics, it is possible that large correlations between samples at widely separated points in the Markov chain could result in apparent bias due to underestimated errors or insufficient thermalization time. We confirmed that this is not the case by comparing against a variety of analytically-known observables. Specifically, we calculated analytically and measured the expectation values of:
\begin{enumerate}
    \item Wilson loops $W_{ab}$, i.e.~traced loops of links of shape $a \times b$;
    \item Polyakov loops $\ell(x) = \tr \left\{ \prod_t U_0(t,x) \right\}$, winding around the periodic boundary of the lattice;
    \item Two-point functions of Polyakov loops, $\ell^*(x) \ell(y)$.
\end{enumerate}
The expectation value of any Polyakov loop is zero due to an exact center symmetry; this result was reproduced by the flow-based samples (as we discuss below, center symmetry is also exact in our models, therefore this quantity is exact based on model proposals even before composition into a Markov chain). Due to confinement, Wilson loops have an expectation value exponentially small in the area of the loop, thus we considered loops of simple shapes up to area 4 and the Polyakov loop two-point function with zero separation, $|\ell(x)|^2$. The flow-based estimates of these quantities for $\SU{2}$ and $\SU{3}$ gauge theory are shown graphically in Fig.~\ref{fig:obs_su2_su3_all}. The results are statistically consistent with the analytical result.

We further checked that as statistics are increased, estimated errors fall as $1/\sqrt{n}$. This must be true asymptotically, but could be modified if there are correlations longer than the finite Markov chain length. We find that errors are indeed consistent with $1/\sqrt{n}$ scaling, as shown for example in Fig.~\ref{fig:err_scaling} for estimates of $\Re W_{11}$ for $\SU{3}$ gauge theory with $\beta = 6$.

\begin{figure}
    \centering
    \includegraphics{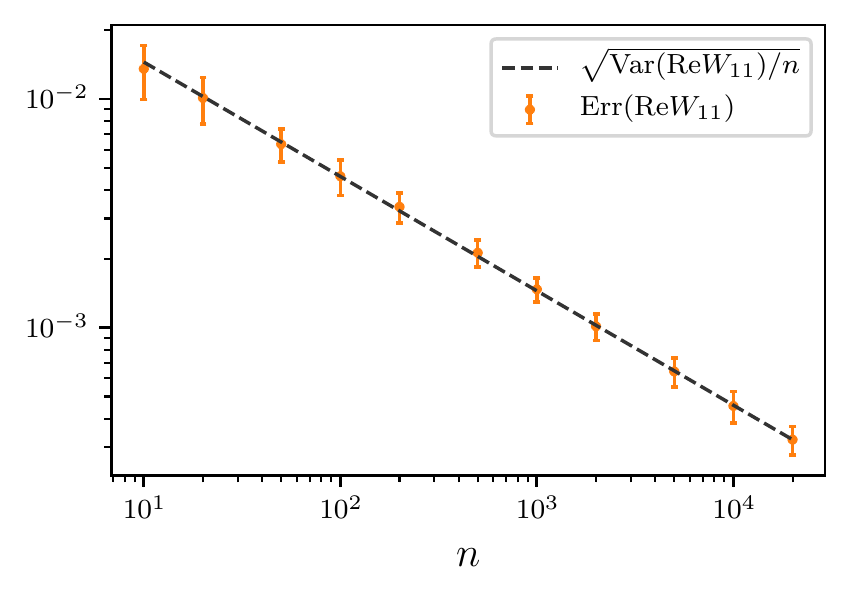}
    \caption{Statistical errors on estimates of $\Re W_{11}$ in $\SU{3}$ gauge theory with $\beta = 6$
    scale as expected as the number of independent samples $n$ is varied. Errors (orange points) are estimated by a bootstrap procedure after thinning the data based on the measured autocorrelation time; the uncertainties on these estimates are measured using an outer bootstrap resampling step. The normalization $\textrm{Var}(\Re W_{11})$ for the theoretical scaling (gray dashed line) is computed using the rightmost measured point.}
    \label{fig:err_scaling}
\end{figure}

\subsection{Symmetries} \label{subsec:application-symms} 
After composition into a Markov chain, flow-based samples are guaranteed to asymptotically reproduce the exact distribution, including all symmetries. However, to reduce Markov chain correlations and improve training efficiency, we constructed our flow-based models to exactly reproduce some symmetries even \emph{when generating proposals}. In terms of the factorization schematically shown in Fig.~\ref{fig:symfact}, exactly imposing symmetries in the model can reduce variance in reweighting factors along the pure-symmetry directions of the distribution.

\begin{figure}
    \centering
    \includegraphics{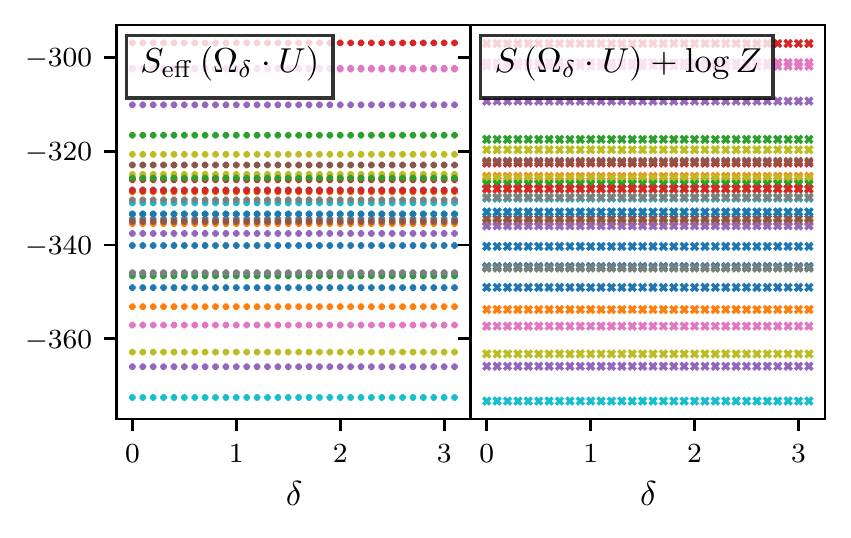}
    \caption{Effective action $\Seff$ vs.~normalized true action $S + \log{Z}$ on a sequence of gauge transformations of 32 gauge configuration samples for $\SU{3}$ gauge theory with $\beta = 6$. The gauge transformation applied is smoothly varied as $\delta$ is increased. Both the flow-based action and true action are exactly invariant to gauge transformations.}
    \label{fig:rugplot-gauge-inv}
\end{figure}

\begin{figure}
    \centering
    \includegraphics{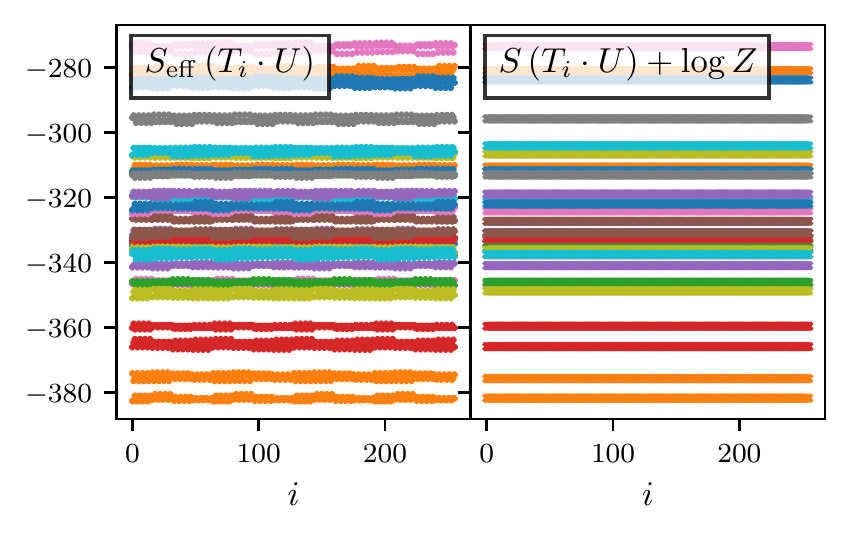}
    \caption{Effective action $\Seff$ vs.~normalized true action $S + \log{Z}$ on a sequence of translations of 32 gauge configuration samples for $\SU{3}$ gauge theory with $\beta = 6$. All $16 \times 16$ translations are plotted in a sequential pattern with index given by $i = \delta y + 16\, \delta x$.}
    \label{fig:rugplot-trans-inv}
\end{figure}

\begin{figure}
    \centering
    \includegraphics{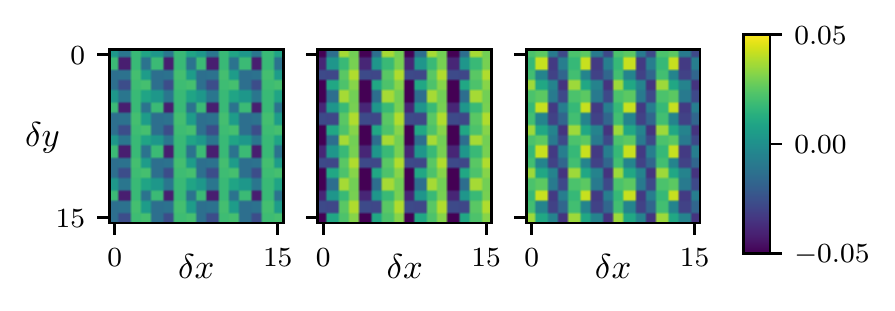}
    \caption{Fluctuations in the flow-based effective action across all possible translations of three random gauge configurations. Configurations are drawn from the model for $\SU{3}$ gauge theory with $\beta = 6$. Fluctuations are reported relative to the mean effective action across all possible translations in each configuration, and are normalized with respect to the standard deviation of the action in the path integral, $\sqrt{\langle \left(S - \langle S \rangle\right)^2 \rangle}$. The action is invariant for shifts $\delta x = 0 \pmod 4, \, \delta y = 0 \pmod 4$ demonstrating the exact translational symmetry modulo $\mathbb{Z}_4 \times \mathbb{Z}_4$.}
    \label{fig:trans-flucts}
\end{figure}

As detailed in Sec.~\ref{subsec:flow-based-symm}, we used coupling layers exactly equivariant to gauge symmetry and translational symmetry modulo a $\mathbb{Z}_4 \times \mathbb{Z}_4$ breaking arising from the size of the tiled pattern. To confirm the exact invariance of the flow-based distribution under gauge transformations, we measured the flow-based effective action over a range of gauge transformations on 32 random configurations along a randomly selected pure-gauge direction. Fig.~\ref{fig:rugplot-gauge-inv} depicts the invariance of both the effective and true actions under this random direction of gauge transformation. The data shown in different colors, corresponding to different random configurations, are approximately aligned in the left and right panels of Fig.~\ref{fig:rugplot-gauge-inv}, indicating that the true action is approximately matched by the effective action in the gauge-invariant directions. We performed a similar investigation of translational invariance by scanning over all $16 \times 16$ possible translations of 32 random configurations. Fig.~\ref{fig:rugplot-trans-inv} shows that there are fluctuations in the flow-based effective action, which arise from symmetry breaking within each $4 \times 4$ tile, but a large subgroup of the translational group is preserved as can be seen by the lines of constant effective action across various translations of each configuration. The spatial structure of the residual fluctuations in the effective action is shown in Fig.~\ref{fig:trans-flucts}. 

We also expect the hypercubic symmetry group to be an exact symmetry in most studies of lattice gauge theories. In the two-dimensional gauge theories under study, this group consists of the 8 possible combinations of rotations and reflections of the lattice. While this symmetry could be imposed in every convolutional neural network used in all coupling layers, the pattern of open loops and the relation of these loops to updated links is difficult to make invariant; choosing a link to ``absorb'' the update of each open loop fundamentally breaks the hypercubic symmetry. On the other hand, this discrete symmetry group has few elements, consisting of only 8 elements in two dimensions, 48 elements in three dimensions, and 384 elements in four dimensions. As such, we instead allowed our flow-based models to learn this small symmetry group over the course of training. Fig.~\ref{fig:rugplot-rr-inv} depicts the approximate invariance of the flow-based effective action on 32 random gauge configurations under all 8 elements of the hypercubic group in 2D.

\begin{figure}
    \centering
    \includegraphics{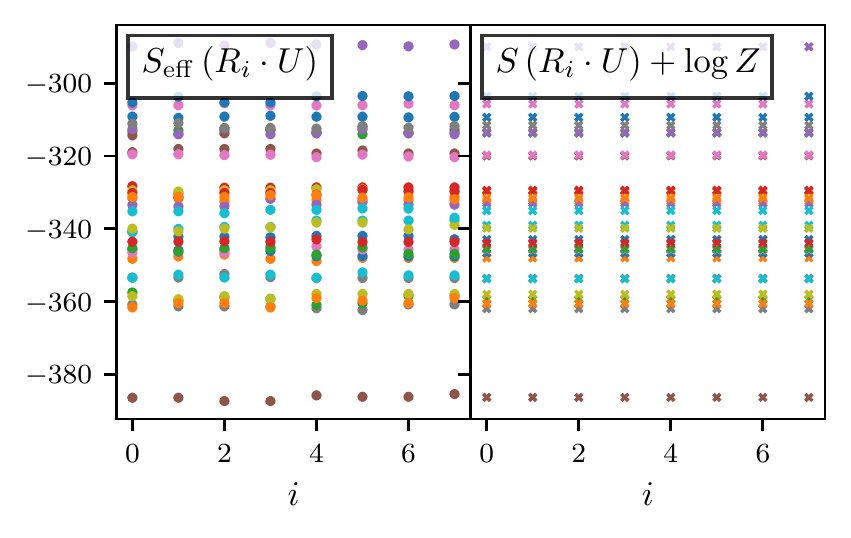}
    \caption{Measured effective action $\Seff$ vs.~normalized true action $S + \log{Z}$ on all 8 possible hypercubic transformations of 32 gauge configurations sampled for $\SU{3}$ gauge theory with $\beta = 6$.}
    \label{fig:rugplot-rr-inv}
\end{figure}

For the pure-gauge theories under consideration, \emph{center symmetry} and \emph{complex conjugation symmetry} are additionally exact symmetries of the action; we included both symmetries explicitly in all of our models. Center symmetry is defined by the transformations
\begin{equation}
    U_0(x) \rightarrow U_0(x) e^{i 2\pi n / N},
    \quad n \in \{0, \dots, N-1\},
\end{equation}
for all links on a fixed time slice, $x_0 = t$, with other links unaffected. Our coupling layers are already equivariant under this symmetry, which can be seen by considering the updated value of any modified link, $U'_\mu(x) = P'_{\mu\nu}(x) P^\dagger_{\mu\nu}(x) U_\mu(x)$: plaquettes do not transform under center symmetry, and by definition center transformations on the link $U_\mu(x)$ are free to commute all the way to the left. If open Polyakov loops were transformed in the coupling layers, or if traced Polyakov loops were used as part of the gauge invariant inputs to any transformation, this property would no longer hold; including terms like these will be necessary for theories in which center symmetry is explicitly broken. We also explicitly constructed our spectral flows to be equivariant under complex conjugation. For $\SU{2}$ matrices, this is equivalent to permutation of the eigenvalues and is therefore immediate. For $\SU{3}$ matrices, it corresponds to a non-trivial mirror symmetry within a single canonical cell. We implemented this mirror symmetry by extending a spline flow from one half of the canonical cell to the entire space using an approach similar to that applied for $\SU{2}$ permutation equivariance. Both center symmetry and complex conjugation symmetry were reproduced to within numerical precision.

\begin{figure}
    \centering
    \includegraphics{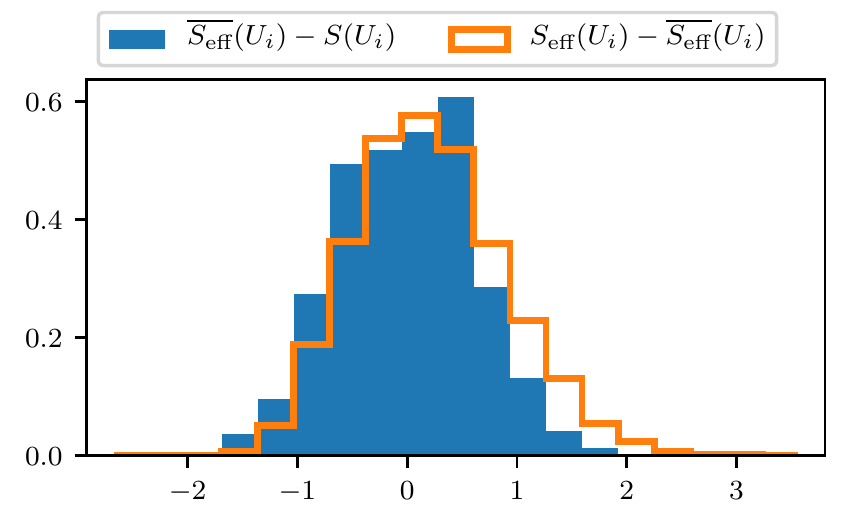}
    \caption{Reweighting factors after post hoc symmetrization (filled blue) vs.~the log difference of the symmetrized effective action from the original effective action (outlined orange). The width of the latter distribution indicates the gains made by averaging over these fluctuations in Eq.~\eqref{eqn:avg-seff}. The width of the former distribution indicates the remaining errors in the model.}
    \label{fig:trans-hist}
\end{figure}

Finally, we considered explicitly symmetrizing model proposals under a discrete symmetry group $H$. Such an approach could be used, for example, to impose the residual $\mathbb{Z}_4 \times \mathbb{Z}_4$ translational symmetry or hypercubic symmetry on the flow-based model post hoc. To do so, a random symmetry transformation is applied after drawing a model proposal and the averaged model weight,
\begin{equation} \label{eqn:avg-seff}
    \overline{\Seff}(U) := -\log\left( \frac{1}{|H|} \sum_{h \in H} e^{-\Seff(h \cdot U)} \right),
\end{equation}
is reported. This averaging over all possible symmetry transformations is required to faithfully report the probability density of the output sample for reweighting or composition into a flow-based Markov chain. It is also very costly if the symmetry group is large (and is intractable for continuous symmetry groups).

We studied the possibility of employing such averaging for the residual $\mathbb{Z}_4 \times \mathbb{Z}_4$ translational symmetry breaking. Fig.~\ref{fig:trans-hist} compares the reweighting factors required for the translationally symmetrized model vs.~the fluctuations that have been averaged over by the sum in Eq.~\eqref{eqn:avg-seff}. The comparable width of these histograms indicates that the improvement in the spread of reweighting factors (which controls the variance of estimators) is $O(1)$; evaluating the ESS directly, we found in this example that the ESS was increased by roughly a factor of two. Thus, the additional factor of 16 in cost required to generate the symmetrized proposals outweighed the variance reduction benefits. We conclude that it is beneficial to impose symmetries when possible in the flow-based model itself, as we did with gauge symmetry, center symmetry, conjugation symmetry, and a large subgroup of translational symmetry, but in our application we found it counterproductive to impose a residual symmetry by brute force averaging of proposals.

\section{Outlook}

It has recently been shown in proof-of-principle work that the challenging computational task of sampling field configurations for lattice gauge theory may be accelerated by orders of magnitude compared with more traditional sampling approaches through the use of flow-based models~\cite{Albergo:2019eim,Kanwar:2020xzo}. 
In other lattice field theories, it has been demonstrated that these models can also be used to estimate observables, such as the absolute value of the free energy, that are difficult to estimate with existing MCMC methods~\cite{Nicoli:2020njz}.

Here, we present a definitive step towards more efficient sampling for lattice gauge theories by 
developing flow-based models that are equivariant under $\SUn$ gauge symmetries,
thus enabling the construction of model architectures that respect the symmetries of the Standard Model of particle and nuclear physics and other physical theories.
We demonstrate the application of this approach to sampling both single $\SUn$ variables and $\SU{2}$ and $\SU{3}$ lattice gauge theory configurations, showing that observables computed using samples from flow-based models are correct within uncertainties and have the predicted statistical scaling with an increasing number of samples.

In the proof-of-principle implementation presented here, we have not attempted to optimize the model architecture and training approaches for expressivity or efficiency. State-of-the-art calculations will likely require further development in these directions. For one, alternative patterns of loops to update in each coupling layer could increase expressivity of the model, and we expect that exploring these choices will have significant impact in higher dimensions, where the degree of connectivity between links and loops is higher. Second, studies of whether the kernels and coupling layers that we constructed can generalize to multimodal distributions will help to understand the ability of these models to capture distributions in broken symmetry phases of lattice gauge theories. Third, investigation of hyperparameter tuning and further ways to exploit existing models for training and model initialization could allow more efficient training and improve model quality. Finally, studying the scaling of model complexity required to take the continuum limit will determine the viability of this approach on the fine-grained lattices employed in state-of-the-art lattice field theory calculations.
Due to locality, keeping the variance of reweighting factors or the flow-based Markov chain rejection rate fixed while we increase the physical volume of the lattice will require improving the model's approximation of the local correlation structure of the theory;\footnote{Instead keeping the model architecture \emph{fixed} while increasing the physical volume results in exponential degradation of the variance of reweighting factors or the flow-based Markov chain rejection rate.}
however, it is not clear how model complexity requirements scale when physical volume is held fixed and the lattice spacing is decreased. This scaling depends on the dynamics of the theory and the architecture of the flow-based model under study, and it must be determined experimentally. If these challenges can be addressed, the extension of these proof-of-principle results to state-of-the-art lattice gauge theory calculations for complex theories such as Quantum Chromodynamics has the potential to redefine the computational limits, and hence the impact, of lattice gauge theory in the coming exascale computing era~\cite{Joo:2019byq}.

\section{Acknowledgements}

The authors thank Jiunn-Wei Chen, Will Detmold, Andreas Kronfeld, George Papamakarios, Andrew Pochinsky, and Peter Wirnsberger for helpful discussions. GK, DB, DCH, and PES are supported
in part by the U.S. Department of Energy, Office of Science, Office of Nuclear Physics, under grant Contract
Number DE-SC0011090. PES is additionally supported by
the National Science Foundation under CAREER Award
1841699 and under EAGER grant 2035015, by the U.S. DOE Early Career Award DE-SC0021006, by a NEC research award, and by the Carl G and Shirley Sontheimer Research Fund. KC is supported by the National Science Foundation under the awards ACI-1450310, OAC-1836650,
and OAC-1841471 and by the Moore-Sloan data science
environment at NYU. MSA is supported by the Carl
Feinberg Fellowship in Theoretical Physics. This work
is associated with an ALCF Aurora Early Science Program project, and used resources of the Argonne Leadership Computing Facility, which is a DOE Office of Science User Facility supported under Contract DE-AC02-06CH11357. Some figures were produced using \texttt{matplotlib}~\cite{Hunter:2007} and \texttt{Mathematica}~\cite{Mathematica}.

\bibliographystyle{apsrev4-1}
\bibliography{main}

\appendix
\section{Proof that equivariance under matrix conjugation can be represented as equivariance under eigenvalue permutation} \label{sec:equiv-conj-perm}
Let $G$ be a compact connected Lie group, such as $\SUn$.
We are interested in characterising the group of diffeomorphisms of $G$ that are equivariant under the action by matrix conjugation. Such a diffeomorphism $f:G\rightarrow G$ satisfies $f(XWX^{-1})=Xf(W)X^{-1}$ for any $W, X\in G$.
Our aim is to show that all such diffeomorphisms are extensions to $G$ of diffeomorphisms of a maximal torus that are equivariant under the action of the Weyl group.

Let $T$ be a maximal torus in $G$.
Recall this torus is equal to its own centraliser $Z(T)=\left\{X\in G|~ XDX^{-1}=D,\forall D\in T\right\}$. 
The Weyl group of $G$ is a finite group equal to $N(T)/T$, where $N(T)=\left\{X\in G|~XDX^{-1}\in T,\forall D\in T\right\}$ is the normalizer of $T$.

In the case of $G=\SUn$ or $G=\Un$, a maximal torus is given by the subgroup of diagonal matrices, and the Weyl group is isomorphic to the group of permutations on $N$ elements acting on $T$ by permuting the elements on the diagonal. For a permutation $\sigma$, a representative matrix in $N(T)$ is given by a permutation matrix, with potentially some elements set to $-1$ instead of $1$ in order for the determinant to be $1$ in the case of $\SUn$.

We start with the easy direction, where we restrict a diffeomorphism from $G$ to $T$.
\begin{prop}
Let $f:G\rightarrow G$ be a matrix conjugation equivariant diffeomorphism. Then $f$ restricted to $T$ is a diffeomorphism of $T$ that is equivariant under the action of the Weyl group.
\end{prop}
\begin{proof}
First, let's show that $f(T)\subset T$. Let $D\in T$. For any $X\in T$, we have $XDX^{-1}=D$ since $T$ is commutative. By equivariance of $f$, we also have $f(XDX^{-1})=Xf(D)X^{-1}$. We deduce that $Xf(D)X^{-1}=f(D)$ for all $X\in T$, which means that $f(D)$ is in the centraliser of $T$. Since this is equal to $T$ for a maximal torus, we have proved $f(D)\in T$.

Since $f$ is a diffeomorphism, its restriction to $T$ is also a diffeomorphism on its image. This image will be both closed and open in $T$, and is therefore the whole of $T$.

Finally, the fact that $f$ restricted to $T$ is equivariant under the action of the Weyl group is immediate, since this action comes from the action by conjugation from $N(T)$.
\end{proof}

For the opposite direction, we restrict ourselves to the cases $G=\SUn$ and $G=\Un$. We choose $T$ to be the subgroup of diagonal matrices.
The Weyl group acts by permutation on the diagonal elements in $T$.

In what follows, we will assume $f:T\rightarrow T$ is a diffeomorphism that is equivariant under the action of the Weyl group.
We first introduce a Lemma that will be used later.
\begin{lemma}\label{lemma:F_commute}
Let $D\in T$. Assume $A\in G$ commutes with $D$, then $A$ also commutes with $f(D)$.
\end{lemma}
\begin{proof}
Let $i,j$ be distinct indices in the range $1\ldots N$. Assume that $D_{ii}=D_{jj}$. We will first prove that $f(D)_{ii}=f(D)_{jj}$.
Let $P\in\SUn$ be given by $P_{ij}=1,~P_{ji}=-1,~P_{ii}=P_{jj}=0$, and $P_{kk}=1$ for all $k\neq i,j$,
then $PDP^{-1}=D$. Since $P\in N(T)$, we have $Pf(D)P^{-1}=f(PDP^{-1})=f(D)$ and $P$ commutes with $f(D)$. This means that $f(D)_{ii}=f(D)_{jj}$.

Let $\lambda_1,\ldots,\lambda_m$ be the $m$ distinct eigenvalues of $D$, with respective multiplicity $n_1,\ldots,n_m$. There exists $P$ in $N(T)$ such that

\begin{equation}
PDP^{-1} = \begin{pmatrix}
\lambda_1 I_{n_1} &  & 0 \\
 & \cdot &  \\
0 &  & \lambda_m I_{n_m}
\end{pmatrix}    
\end{equation}
where $I_{n_k}$ is an identity matrix of size $n_k$. This means that $f(PDP^{-1})$ must also be of the form
\begin{equation}
f(PDP^{-1}) = \begin{pmatrix}
\mu_1 I_{n_1} &  & 0 \\
 & \cdot &  \\
0 &  & \mu_m I_{n_m}
\end{pmatrix}    
\end{equation}

Since $A$ commutes with $D$, we have that $PAP^{-1}$ commutes with $PDP^{-1}$. Since matrices that commute must preserve each others eigenspaces,
this implies that $PAP^{-1}$ must have the form
\begin{equation}
PAP^{-1} = \begin{pmatrix}
U_1 & & 0 \\
 & . & \\
0 & & U_n
\end{pmatrix}
\end{equation}
Given the form of $Pf(D)P^{-1}=f(PDP^{-1})$ shown above, we conclude that $PAP^{-1}$ commutes with $Pf(D)P^{-1}$, therefore $A$ commutes with $f(D)$.
\end{proof}

Finally, using Lemma~\ref{lemma:F_commute}, we can prove our main result.
\begin{prop}\label{prop:F_well_defined}
Assume $W\in G$ has two different decompositions $W=XDX^{-1}=YEY^{-1}$, where $D$ and $E$ are diagonal matrices. Then
\begin{equation}\label{eq:F_well_defined}
    Xf(D)X^{-1} = Yf(E)Y^{-1}.
\end{equation}
\end{prop}
\begin{proof}
There exists $P$ in $N(T)$ such that $E=PDP^{-1}$, which implies $f(E)=Pf(D)P^{-1}$. This means \Cref{eq:F_well_defined} is equivalent to
\begin{equation}\label{eq:F_well_defined_D}
    Xf(D)X^{-1} = Zf(D)Z^{-1},
\end{equation}
where $Z=YP$. The above equation is equivalent to saying that $X^{-1}Z$ commutes with $f(D)$, and by \Cref{lemma:F_commute} this will be the case if $X^{-1}Z$ commutes with $D$. This is easy to prove:
\begin{equation}
\begin{split}
    X^{-1}ZDZ^{-1}X & = X^{-1}YPDP^{-1}Y^{-1}X \\
    & = X^{-1}YEY^{-1}X\\
    & = X^{-1}WX \\
    & = D
\end{split}
\end{equation}
\end{proof}

\begin{ex}
In the case of $G=\SU{2}$, the maximal torus is isomorphic to $\U{1}$, the Weyl group has size $2$ and its only non-trivial element transforms
$\left(\begin{smallmatrix} \lambda & 0 \\ 0 & \overline{\lambda} \end{smallmatrix}\right)$
to
$\left(\begin{smallmatrix} \overline{\lambda} & 0 \\ 0 & \lambda \end{smallmatrix}\right)$, thus any bijection $f:\U{1}\rightarrow \U{1}$ that satisfies $f(\overline{z})=\overline{f(z)}$ defines an equivariant bijection of $\SU{2}$.
\end{ex}

According to proposition \Cref{prop:F_well_defined}, any matrix conjugation equivariant function on $T$ can be extended to an equivariant function on $G$.
If the function was invertible on $T$, then it is easy to see that it will also be invertible on $G$.

\section{Details of permutation equivariance of $\SUn$ spectral flows} \label{sec:canon-alg}

\subsection{Proof that Eq.~\ref{eq:canon-cell-vert} defines a cell}
\label{sec:cell-alg-proof}

We demonstrate that the vertices from Eq.~\eqref{eq:canon-cell-vert} define an $(N-1)$-simplex $\Psi$ corresponding to a cell $\mathcal{C}$. In practice, this means showing that any point on the boundary of $\Psi$ maps to a point in $\mathcal{C}$ with repeated eigenvalues, while any point in the interior of $\Psi$ maps to a regular point, i.e. one without repeated eigenvalues.

\begin{prop}
\label{prop:canonchamber_v2}
 The vertices $y_1,\ldots,y_N$ from Eq.~\eqref{eq:canon-cell-vert} define an $(N-1)$-simplex $\Psi$ that maps to a cell $\mathcal{C}=\mathrm{exp}(\Psi)$ in the maximal torus.
\end{prop}
\begin{proof}
Let $\theta$ be a point in $\Psi$, the convex hull of $y_1,\ldots,y_N$, given by
\[ \theta_j = 2 \pi \sum_k \gamma_k \left( \frac{k}{N} -
    \delta_{k \geqslant j} \right), \]
where $\gamma_k \geqslant 0$ and $\sum_k \gamma_k = 1$. The boundary $\partial \Psi$ is the simplicial complex formed by all points $\theta$ such that at least one $\gamma_k$ is zero.

We consider the difference between two points $\theta_i$ and $\theta_j$, for $j > i$,
\begin{align}
 \theta_j - \theta_i &= 2 \pi \sum_k \gamma_k \left( \delta_{k \geqslant
   i} - \delta_{k
   \geqslant j} \right)\\ 
   &= 2 \pi \sum_{k = i}^{j - 1}
   \gamma_k .\label{eq:thetai_minus_thetaj}
\end{align}
If $\gamma_k = 0$ for some $k$ such that $1 \leqslant k < N$, then $\theta_{k+1} - \theta_{k} = 0$ and $\textrm{exp}(\theta)$ has a repeated eigenvalue. If
$\gamma_N = 0$, we have that $\theta_N - \theta_1 = 2 \pi\sum_{k = 1}^{N - 1} \gamma_k$, but since $\gamma_N = 0$ we have $\sum_{k = 1}^{N- 1} \gamma_k = 1$. Thus
$
 \theta_N - \theta_1 = 2 \pi
$.
This shows that a point in $\partial\Psi$ is exponentiated to a point with repeated eigenvalues.

Finally, we need to show that a point in the interior of $\Psi$ is exponentiated to a regular point. Such a point corresponds to $\gamma_k>0,\forall k$. As a consequence of Eq.~\eqref{eq:thetai_minus_thetaj}, no pair $\theta_i,\theta_j$ are equivalent modulo $2\pi$: for an interior point, the sum $\sum_{k=1}^{j-1}\gamma_k$ is strictly positive and also strictly smaller than $1$.
\end{proof}

\subsection{More on Eq.~\ref{eq:canon-cell-vert}}
\label{sec:cell-alg-deriv}
In this section, we explain where the points $y_k$ in Eq.~\eqref{eq:canon-cell-vert} come from. In particular, we draw some parallel between our construction of the simplex $\Psi$ and fundamental domains in Bravais lattices. Fig.~\ref{fig:steps_chambers} depicts the steps described below.

Recall that we defined a cell in section~\ref{subsec:sun-flows} as the closure of a connected component in $T$ of the set of regular points. These cells are separated by regions where the eigenvalues $\lambda_i$ are
degenerated. That is, for every point $(\lambda_1,\ldots,\lambda_N)$ in the boundary of a cell, there exists at least one
pair $i, j \in \{ 1, \ldots, N \}$ such that $\theta_i - \theta_j = 0
\mod 2 \pi$, where $\theta_i = \arg \lambda_i$. The inverse image of these boundaries under the exponential map are affine hyperplanes in $\torusalg$ that bound simplexes.
The set of vertices of these simplexes (and all their translated copies) form a Bravais
lattice $\Lambda$ given by $\sum_i z_i u_i$,where $z \in
\mathbb{Z}^{N - 1}$ and the primitive vectors $u_i \in \mathbb{R}^{N}$ are defined by $[u_i]_k = 2 \pi (\frac{1}{N} - \delta_{(i + 1) k})$. In the
theory of lattices, the canonical cell is known as the fundamental simplex and its orbit under the Weyl group is known as the root polytope~\cite{koca2018explicit}.

\begin{figure}[t!]
    \centering
    \includegraphics[width=0.95\linewidth]{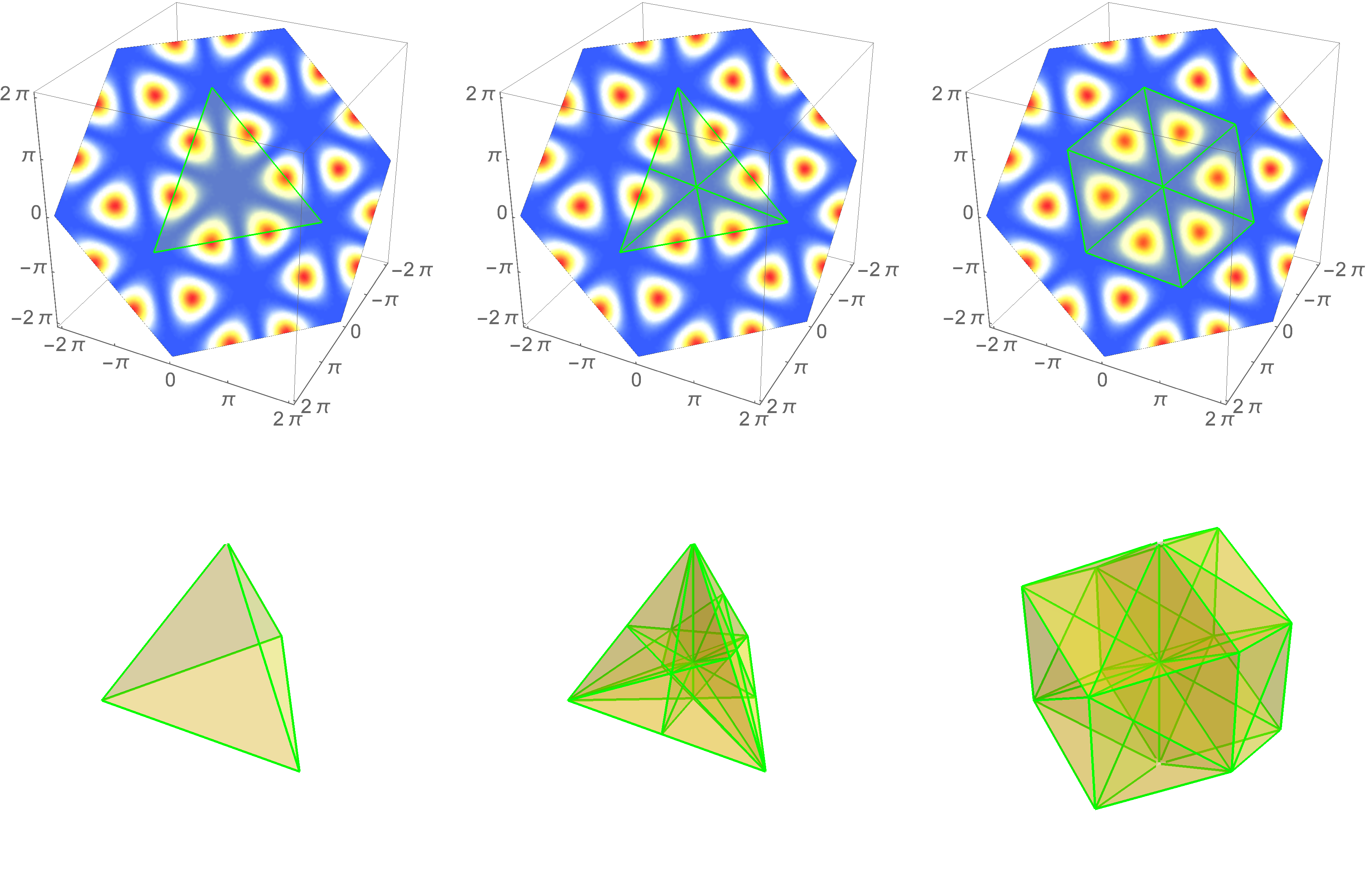}
    \caption{Illustration of the steps to derive Eq.~\eqref{eq:canon-cell-vert} for $\SU{3}$ [top row] and $\SU{4}$ [bottom row]. Columns correspond, from left to right, to the steps: (i) Starting simplex $\Delta$ on the hyperplane $\sum_j \theta_j = 0$; (ii) Barycentric subdivison and (iii) Edge-length adjustment to cover the cell. For $\SU{3}$ we overlay the simplexes on top of the Haar measure for reference.}
    \label{fig:steps_chambers}
\end{figure}

Recall that we identified $\torusalg$ with the hyperplane $\sum_i \theta_i = 0$ in $\mathbb{R}^N$.
Assume $\Delta$ is an $(N - 1)$-simplex with $N$
vertices $q_i \in \torusalg$, defined by the components
\[ [q_i]_j = 2 \pi \left( \frac{1}{N} - \delta_{i j} \right), \]
where $[q_i]_j$ indicates the $j$th component of the $i$th vertex.

Based on the observation that Weyl chambers in the Lie algebra of $\SUn$ are open cones determined by the barycentric subdivision of $\Delta$~\cite{guillemin1996symplecticChap5}, 
we then apply a barycentric transformation to the vertices $q_i$. Generally, a barycentric transformation is achieved by sending $\{v_k\}$ to $\{w_k\}$ by
\begin{equation}
    w_k = \frac{1}{k}\sum_{i=1}^k v_i.
\end{equation}
Applying this to $\Delta$ gives a new simplex $\tilde{\Delta}$ with vertices $x_k$:
\[ [x_k]_j = \frac{2 \pi}{k} \sum_{i = 1}^k \left( \frac{1}{N} -
    \delta_{i j} \right) = \frac{2 \pi}{k} \left( \frac{k}{N} - \delta_{k \geqslant j} \right), \]
where $\delta_{k \geqslant j}$ is equal to one when $k \geqslant j$ and
zero otherwise.

The simplex $\tilde{\Delta}$ is correctly aligned with a cell, however it does not contain the full cell. To fix this, we adjust the length of the edges of
$\tilde{{\Delta}}$, resulting in the final simplex $\mathcal{C}$. Let
$\alpha_k \in \mathbb{R}^+$ be an arbitrary scaling of the $k$th edge of
$\tilde{{\Delta}}$, so that the rescaled coordinates of the vertices are
given by
\[ [y_k]_j = \alpha_k \frac{2 \pi}{k} \left( \frac{k}{N} -
    \delta_{k \geqslant j} \right) . \]
We want the vertices $\{ y_1, \ldots, y_N \}$ of $\mathcal{C}$ to
correspond to neighbouring points of the Bravais lattice $\Lambda$
, so that the
orbit of the simplex $\mathcal{C}$ with respect to the Weyl group tiles
the torus without leaving holes and without overlaps. More precisely, this constraint is
equivalent to imposing $[y_k]_j - [y_k]_i \in \{0, 2\pi\}, \forall j >
i$. Substituting the expression for $[y_k]_{j}$ we obtain the constraint
\[ [y_k]_j - [y_k]_i = \alpha_k \frac{2 \pi}{k} (\delta_{k \geqslant i} -
   \delta_{k \geqslant j}) = 0 \mod 2 \pi, \forall j > i. \]
The term $\delta_{k \geqslant i} - \delta_{k \geqslant j}$ can have values
0 or 1, so that the constraint can be satisfied for all $j > i$ if
$\alpha_k = k$.
With this choice, the coordinates of the vertices of $\mathcal{C}$ are
given by
\begin{align}
    [y_k]_j &= 2 \pi \left( \frac{k}{N} - \delta_{k \geqslant j}
\right).
\end{align}

\subsection{Proof that Algorithm~\ref{alg:canoncell}
projects into $\Psi$}\label{sec:proof_algo1}
In this section, we will show the output of Algorithm~\ref{alg:canoncell} is always a point in $\Psi$.

The output of Algorithm~\ref{alg:canoncell} is a set of angles $\theta^{\rm canon}$. In this section, we will write $\theta^c$ as an abbreviation for $\theta^{\rm canon}$. We wish to prove that $\theta^c$ is in the convex hull of the $y_k$ defined in Eq.~\eqref{eq:canon-cell-vert}. We will do so by explicitly exhibiting the weights of the convex combination. In essence, our proof is the opposite of what lead to Eq.~\eqref{eq:thetai_minus_thetaj}.

Define
\begin{align}
\gamma_k &= \left\{ \begin{array}{cc}
  \frac{1}{2\pi}(\theta^c_{k+1}-\theta^c_k) & k<N,\\
  1-\sum_{j=1}^{N-1}\gamma_j & k=N.
\end{array} \right.
\end{align}
The sum $\sum_{j=1}^{N-1}\gamma_j$ simplifies to $\frac{1}{2\pi}(\theta^c_N-\theta^c_1)$. By construction, the difference $\theta^c_N-\theta^c_1$ cannot be more than $2\pi$. It follows that $\gamma_k\geq 0,\forall k$ and $\sum_k\gamma_k=1$.

Let $\theta'=\sum_k\gamma_k y_k$ be in $\Psi$. We will now prove that $\theta'=\theta^c$, which will conclude the proof. Using the definition of $y_k$ in Eq.~\eqref{eq:canon-cell-vert}, it follows that:
\begin{equation}
    \begin{split}
        \theta'_j & = 2\pi\sum_k\gamma_k \left(\frac{k}{N}-\delta_{k\geq j} \right) \\
        & = 2\pi\left(\sum_{k=1}^{N-1}\frac{1}{2\pi}(\theta^c_{k+1}-\theta^c_k)(\frac{k}{N}-\delta_{k\geq j})\right) + 2\pi\gamma_N [y_N]_j.
    \end{split}
\end{equation}
The extra term after the initial sum above is $0$ because $[y_N]_j=0$. We continue:
\begin{equation}
    \begin{split}
        \theta'_j & = \sum_{k=1}^{N-1}(\theta^c_{k+1}-\theta^c_k)(\frac{k}{N}-\delta_{k\geq j}) \\
        & = \sum_{k=2}^N\theta^c_k(\frac{k-1}{N}-\delta_{k-1\geq j}) - \sum_{k=1}^{N-1}\theta^c_k(\frac{k}{N}-\delta_{k\geq j}) \\
        & = \theta^c_N(\frac{N-1}{N}-\delta_{N-1\geq j}) + \sum_{k=2}^{N-1}\theta^c_k(-\frac{1}{N}+\delta_{k\geq j} - \delta_{k\geq j+1})\\
        &~~~-\theta_1^c(\frac{1}{N}-\delta_{1\geq j})
    \end{split}
\end{equation}
In the last line above, note that we can simplify $\frac{N-1}{N}-\delta_{N\geq 1+j}$ to $\frac{-1}{N} + \delta_{j,N}$, and $\delta_{k\geq j} - \delta_{k\geq j+1}$ to $\delta_{k,j}$, and also $\delta_{1\geq j}$ to $\delta_{1,j}$. It follows that
\begin{equation}
    \begin{split}
        \theta'_j & = \theta^c_N\delta_{N,j} + \sum_{k=2}^N\theta_k^c\delta_{k,j} + \theta^c_1\delta_{1,j} -\frac{1}{N}\sum_{k=1}^N\theta^c_k\\
        & = \theta^c_j
    \end{split}
\end{equation}
The last line above was obtained using that the sum of $\theta_k^c$ is $0$. This concludes our proof.

\subsection{Full algorithm}

\begin{algorithm}[H]
\caption{Equivariant $\SUn$ coupling layer \label{alg:full}}
Given $U \in \SUn$
\begin{enumerate}
\item $\lambda, \{\vec{v}_i\} = \textrm{eigendecomp}(U)$
\item Project to canonical cell $\Psi$: $I = \textrm{canonicalize}(\textrm{arg}(\lambda))$
\item Map to axis-aligned simplex $\Delta$: $\beta = \zeta^{-1}(I)$
\item Map to box $\Omega$: $\alpha = \phi^{-1}(\beta)$
\item Apply box flow: $\alpha' = \chi(\alpha)$
\item $\beta' = \phi(\alpha')$
\item $I' = \zeta(\beta')$
\item $\lambda' = \textrm{uncanonicalize}(I')$
\item $U' = \textrm{eigenrecomp}(\lambda', \{\vec{v}_i\})$
\item Accumulate all log-det-Jacobians:\\
\begin{equation*}
\begin{split}
\textrm{LDJ} = & \log \textrm{Haar}(\lambda') - \log \textrm{Haar}(\lambda) +\\
& \textrm{LDJ}_{\chi} + \textrm{LDJ}_{\phi^{-1}} - \textrm{LDJ}_{\phi}
\end{split}
\end{equation*}
\item $U'$ is equivariant to $\SUn$ matrix conjugations and $\textrm{LDJ}$ is invariant to matrix conjugations.
\item return $U'$ and $\textrm{LDJ}$

\end{enumerate}
\end{algorithm}
Above, there are no terms in $\LDJ$ for the map $\zeta$ because the Jacobian factor acquired from the forward and backward maps are constants that cancel.
The term $\textrm{Haar}(\lambda)$ gives the density of the Haar measure with respect to the Lebesgue measure in the space of eigenvalues, as defined in Eq.~\eqref{eqn:haar-density}. The normalization of this term is unimportant as it cancels in the above algorithm.

\section{Backpropagation through unitary matrix diagonalization} \label{sec:backprop-diag}
We define the backpropagation of gradients through application of a black-box (unitary) diagonalization operation on unitary matrices, i.e.~the steps required to produce a gradient of a scalar loss function $L$ with respect to the matrices, given the gradient of $L$ with respect to their eigenvalues and (unit-norm) eigenvectors. It is assumed that the loss function $L$ is independent of the details of the diagonalization procedure, including the overall complex phase of each eigenvector and the permutation of eigenvalues and eigenvectors; this assumption is true for our spectral flows, for example. A gradient backpropagation algorithm suitable for a black-box diagonalization procedure allows us to implement the diagonalization using any approach that is efficient and numerically stable.

Given the $N \times N$ unitary matrix $U$, we define the eigenvalues and eigenvectors returned by the black-box diagonalization step to be $w = (w_1, \dots, w_N)$ and $P = (\vec{v}_1, \dots, \vec{v}_N)$, respectively. By definition of unitary diagonalization, they satisfy
\begin{equation} \label{eqn:diag}
    U = P D P^\dagger, \quad D := \diag(w).
\end{equation}
Eq.~\eqref{eqn:diag} does not fully constrain $d$ and $P$, so they may further depend on $U$; such dependence (e.g.~how the overall phases on each $\vec{v}_i$ are chosen) is an implementation detail of the diagonalization procedure.
We define the vector of gradients given as input to the backpropagation step to be
\begin{equation}
  g := \left( \frac{\partial L}{\partial \Re w}, \frac{\partial
    L}{\partial \Im w}, \frac{\partial L}{\partial \Re P}, 
  \frac{\partial L}{\partial \Im P} \right),
\end{equation}
where we implicitly bundle the components of the gradients with respect to $w$ and $P$ into one row vector.\footnote{Note that machine learning libraries with support for complex numbers may provide such gradients in different formats. For example, the convention used by \texttt{JAX}~\cite{jax2018github} is to provide the complex-valued gradient vector
\begin{equation}
  g^{\texttt{jax}} = \left( \frac{\partial L}{\partial \Re w} - i \frac{\partial
    L}{\partial \Im w}, \frac{\partial L}{\partial \Re P} - i
  \frac{\partial L}{\partial \Im P} \right),
\end{equation}
which packs the gradient components into complex values, matching the type of $w$ and $P$~\cite{maclaurin2016autograd}.
}

To proceed, we use Eq.~\eqref{eqn:diag} to relate the differential elements $dP$ and $dw$ to $dU$, and ultimately solve for the Jacobian elements $\partial \Re w / \partial \Re U$, $\partial \Im w / \partial \Re U$, \dots, $\partial \Im P / \partial \Im U$. Ambiguity due to the implementation details of the diagonalization procedure correspond to ambiguities in components of the Jacobian that cannot affect $L$ by our assumption above. Therefore, in what follows we simply make a valid choice.

From the unitarity of $P$ and diagonal nature of $D$, we know
\begin{equation} \label{eqn:Puni}
\begin{gathered}
P dP^\dagger + dP P^\dagger = 0, \quad
dP^\dagger P + P^\dagger dP = 0, \\
dw_i = dD_{ii}, \quad \text{and} \quad d D_{ij} = 0, \forall i \neq j.
\end{gathered}
\end{equation}
From Eq.~\eqref{eqn:diag}, we can derive
\begin{equation}
\begin{aligned}
U = P D P^\dagger \implies P^\dagger U P &= D \\
dP^\dagger U P + P^\dagger dU P + P^\dagger U dP &= dD \\
dP^\dagger P D + P^\dagger dU P + D P^\dagger dP &= dD.
\end{aligned}
\end{equation}
Introducing
\begin{equation}
dH := P^\dagger dP = -dP^\dagger P,
\end{equation}
which represents the differential of $P$ translated to the identity, we simplify the relation of differential elements to
\begin{equation}
P^\dagger dU P - dD = -D dH + dH D = [dH, D].
\end{equation}
However, since $D$ is diagonal we know $[dH, D]_{ii} = 0$ and therefore have an explicit expression for $dw_i$
\begin{equation}
    dw_i = dD_{ii} = [ P^\dagger dU P ]_{ii}.
\end{equation}
We can similarly compute the off-diagonal components of $dH$,
\begin{equation}
\begin{aligned}
    [P^\dagger dU P]_{ij} &= (w_j - w_i) dH_{ij}, \forall i \neq j \\
    dH_{ij} &= \frac{[P^\dagger dU P]_{ij}}{w_j - w_i}.
\end{aligned}
\end{equation}
The imaginary components of the diagonal elements of $dH$ are unconstrained ($dH$ is anti-Hermitian so the real components are zero), reflecting the fact that the only undefined (continuous) degrees of freedom are the phases on the eigenvectors. We are free to set them to zero, giving a \emph{valid} choice of $dP$,
\begin{equation}
\begin{gathered}
    dP_{mn} = [P dH]_{mn} = P_{mi} P^\dagger_{ij} dU_{jk} P_{kn} V_{in}, \quad \text{where} \\
    V_{in} = \begin{cases}
        \frac{1}{w_n - w_i} & i \neq n \\
        0 & \text{else}.
    \end{cases}
\end{gathered}
\end{equation}

Having defined a valid solution of the differentials $dP$ and $dw$ in terms of $dU$, we can solve for all the components of the Jacobian:
\begin{equation}
\begin{aligned}
    \frac{\partial \Re w_i}{\partial \Re U_{jk}} &= \Re( P^*_{ji} P_{ki} ) \\
    \frac{\partial \Im w_i}{\partial \Re U_{jk}} &= \Im( P^*_{ji} P_{ki} ) \\
    \frac{\partial \Re w_i}{\partial \Im U_{jk}} &= -\Im( P^*_{ji} P_{ki} ) \\
    \frac{\partial \Im w_i}{\partial \Im U_{jk}} &= \Re( P^*_{ji} P_{ki} ) \\
    \frac{\partial \Re P_{mn}}{\partial \Re U_{jk}} &= \Re( P_{mi} P^*_{ji} P_{kn} V_{in} ) \\
    \frac{\partial \Im P_{mn}}{\partial \Re U_{jk}} &= \Im( P_{mi} P^*_{ji} P_{kn} V_{in} ) \\
    \frac{\partial \Re P_{mn}}{\partial \Im U_{jk}} &= -\Im( P_{mi} P^*_{ji} P_{kn} V_{in} ) \\
    \frac{\partial \Im P_{mn}}{\partial \Im U_{jk}} &= \Re( P_{mi} P^*_{ji} P_{kn} V_{in} ). \\
\end{aligned}
\end{equation}
Together these components form the Jacobian matrix
\begin{equation}
    J = \left( \begin{matrix}
    \dfrac{\partial \Re w}{\partial \Re U} & \dfrac{\partial \Re w}{\partial \Im U} \\[1em]
    \dfrac{\partial \Im w}{\partial \Re U} & \dfrac{\partial \Im w}{\partial \Im U} \\[1em]
    \dfrac{\partial \Re P}{\partial \Re U} & \dfrac{\partial \Re P}{\partial \Re U} \\[1em]
    \dfrac{\partial \Im P}{\partial \Re U} & \dfrac{\partial \Im P}{\partial \Re U}
    \end{matrix} \right),
\end{equation}
which allows us to backpropagate the gradients by right-multiplication,
\begin{equation}
    g' := \left( \frac{\partial L}{\partial \Re{U}}, \frac{\partial L}{\partial \Im{U}} \right) = g J.
\end{equation}

\section{Conjugation equivariant maps on $\SU{3}$ via averaging} \label{sec:su3-equiv-by-averaging}
We are interested in building diffeomorphisms of $\SU{3}$ that are equivariant under the action by conjugation of $\SU{3}$ on itself. We already know from Appendix~\ref{sec:equiv-conj-perm} that it is enough to build diffeomorphisms of $T$ that are equivariant under the action of its Weyl group. Our goal here is to tackle that problem by lifting it to the Lie algebra $\torusalg$ of $T$, where we will \textit{average} diffeomorphisms of $\torusalg$ to force the equivariance. We will then identify certain sufficient properties that guarantee this still leads to diffeomorphisms of $\torusalg$ that descend to diffeomorphisms of $T$.

In this section, we identify the Lie algebra of $T$ with $\mathbb{R}^2$ via the map
\begin{equation}
    (x, y)\overset{\exp}{\longmapsto} \textrm{Diag}(e^{2\pi i x}, e^{2\pi i y}, e^{-2\pi i (x+y)}).
\end{equation}
Given a map $\tilde{h}:\torusalg\rightarrow\torusalg$, we say that this map descends to a map on $T$ if there exists $h$ such that the following diagram is commutative
\begin{equation}
    \begin{tikzcd}
        \torusalg \arrow{r}{\tilde{h}} \arrow[swap]{d}{\exp} & \torusalg \arrow{d}{\exp} \\
        T \arrow{r}{h}& T
    \end{tikzcd}
\end{equation}
If $\tilde{h}$ is equivariant with respect to the action of the Weyl group, then so is $h$ since $exp$ is equivariant and surjective.
Given $(x, y)\in \mathbb{R}^2$, we define $z=-x-y$.

The Weyl group $\mathcal{W}$ associated with $T$ is the group of permutations over $3$ elements. This group acts on both $T$ and its Lie algebra $\torusalg$, and the exponential map is equivariant under these actions.
Denote $\sigma_0,\ldots,\sigma_5$ the elements of $\mathcal{W}$.

A map $\tilde{h}:\torusalg\rightarrow\torusalg$ descends to a map $h:T\rightarrow T$ iff
\begin{equation}\label{eq:descends_to_torus}
    \forall x, y\in\torusalg, \forall a,b\in\mathbb{Z}^2, \tilde{h}(x+a, y+b) - \tilde{h}(x, y)\in \mathbb{Z}^2.
\end{equation}
Also, $\tilde{h}:\torusalg\rightarrow\torusalg$ if a local diffeomorphism iff its descended map $T\rightarrow T$ is a local diffeomorphism.

\begin{prop}
Let $\tilde{h}:\torusalg\rightarrow\torusalg$ be any map that satisfies \Cref{eq:descends_to_torus}, then
\begin{equation}\label{eq:average_diff_su3}
    G_{\tilde{h}} = \frac{1}{6}\sum_k \sigma_k^{-1}\tilde{h}\sigma_k
\end{equation}
also satisfies \Cref{eq:descends_to_torus}, and is equivariant under the action of the Weyl group.
\end{prop}
\begin{proof}
The set of maps that satisfy \Cref{eq:descends_to_torus} is stable under convex combination and composition. It follows that $G_{\tilde{h}}$ satisfies \Cref{eq:descends_to_torus}.

Let's check it is equivariant. Let $\sigma_i$ be in $\mathcal{W}$. In particular, $\sigma_i$ is an affine map, and it will preserve barycenters, so that
\begin{equation}
    \begin{split}
        \sigma_i\circ G_{\tilde{h}} & = \frac{1}{6}\sum_k \sigma_i\sigma_k^{-1}\tilde{h}\sigma_k \\
        & = \frac{1}{6}\sum_k (\sigma_k\sigma_i^{-1})^{-1}\tilde{h}\sigma_k \\
        & = \frac{1}{6}\sum_k \sigma_k^{-1}\tilde{h}\sigma_k\sigma_i \\
        & = G_{\tilde{h}}\circ \sigma_i
    \end{split}
\end{equation}
\end{proof}

Note that if $\tilde{h}$ is the identity map, then $G_{\tilde{h}}$ is also the identity map. If we start with a $\tilde{h}$ that is close to the identity, we will have constructed an equivariant diffeomorphism of $\torusalg$, that descends to an equivariant diffeomorphism of $T$.

In order to ensure that $G_{\tilde{h}}$ is a diffeomorphism, we will now restrict $\tilde{h}$ to a particular form. Namely, assume $f:\mathbb{R}\rightarrow\mathbb{R}$ is a diffeomorphism obtained by lifting a diffeomorphism of $S^1$ to its Lie algebra $\mathbb{R}$, and
define $\tilde{h}(x, y)=(f(x), f(y))$. This is a diffeomorphism of $\torusalg$, and $G_{\tilde{h}}$ is equivariant and descends to an equivariant map $G_h:T\rightarrow T$. We would like to find sufficient conditions for $G_{\tilde{h}}$ to descend to a diffeomorphism of $T$.

Let's start by computing the Jacobian of $G_{\tilde{h}}$.
\begin{prop}\label{prop:jacobian}
The Jacobian $J(G)$ of $G_{\tilde{h}}$ is given by
\begin{equation}
    J(G)(x, y) = \frac{1}{3}\begin{bmatrix}
2f'(x) + f'(z) & f'(z) - f'(y)\\
f'(z) - f'(x) & 2f'(y) + f'(z)
\end{bmatrix}	
\end{equation}
\end{prop}
\begin{proof}
This is a direct, albeit a bit tedious, computation using the Jacobians of the elements of the Weyl group.
\end{proof}

\begin{coro}\label{cor:g_h_diff}
The map $G_{\tilde{h}}$ is a diffeomorphism of $\torusalg$.
\end{coro}
\begin{proof}
Let's start by checking that $G_{\tilde{h}}$ is a local diffeomorphism. We only need to check that its Jacobian is always invertible. The determinant of the Jacobian in \Cref{prop:jacobian} simplifies nicely to $\frac{1}{3}(f'(x)f'(y) + f'(y)f'(z) + f'(z)f'(x))$. Since we assumed that $f$ comes from a diffeomorphism of $S^1$, its derivative is either always strictly positive, or always strictly negative, and the determinant cannot vanish.

Since $\torusalg$ is simply connected, this means $G_{\tilde{h}}$ is indeed a diffeomorphism of $\torusalg$ that satisfies \Cref{eq:descends_to_torus}.
\end{proof}

\begin{coro}
If $f$ is connected to the identity by a path of diffeomorphisms, then $G_{\tilde{h}}$ descends to a diffeomorphism $G_h$ of $T$ that is equivariant under the Weyl group.
\end{coro}
\begin{proof}
We already know that $G_{\tilde{h}}$ descends to a local diffeomorphism of $T$. This is necessarily a covering of $T$ by itself. We only need to prove this covering is trivial. We cannot use the same argument with $T$ as we did with $\torusalg$, because $T$ is not simply connected.

We have assumed that $f$ is homotopic to the identity of $\mathbb{R}$. This immediately gives us a homotopy from $G_{\tilde{h}}$ to the identity of $\torusalg$. This descends to a map $k:[0,1]\times T\rightarrow T$. Using local coordinates, we can see that this map is continuous. It therefore defines a homotopy from $G_h$ to the identity of $T$. In particular, we conclude that $G_h$ must induce the identity map on the fundamental group, and is necessarily a trivial covering.

\end{proof}

\begin{coro}\label{cor:su3_diff}
Any circle diffeomorphism from Ref.~\cite{papamakarios2019normalizing}, such as a mixture of NCPs, M\"obius, or a spline, can be used to define an equivariant diffeomorphism of $\SU{3}$.
\end{coro}

We tested flows based on the equivariant diffeomorphisms suggested by Corollary~\ref{cor:su3_diff} but found that networks built this way did not perform as well as those used in the main body of the paper. This is likely because using a single circle diffeomorphism in Equation~\eqref{eq:average_diff_su3} is too restrictive. An alternative would be to build a diffeomorphism of the torus from two circle diffeomorphisms by auto-regressivity. In that case, Corollary~\ref{cor:g_h_diff} does not apply and one needs to be careful that averaging still leads to a diffeomorphism.

\section{The case of $\Un$}\label{sec:case_of_un}
The case of $\Un$ is simpler than $\SUn$ because we do not have the constraint that the determinant must be equal to $1$. We could apply the same strategy used for the $\SUn$ flows via a canonicalization map to map every point to a canonical cell, and then build a flow in the $N$-simplex cell (in contrast to the $(N-1)$-simplex cell for $\SUn$). An alternative and simpler direction is to directly build a permutation equivariant flow on the torus $T^N$. This can be achieved by first mapping $T^N$ to $\mathbb{R}^N$ using a non-compact projection~\cite{gemici2016normalizing,papamakarios2019normalizing}, then stacking layers alternating between those defined by Eq.~(13) and Eq.~(15) or Eq.~(16) in Ref.~\cite{bender2019exchangeable}, before finally projecting back to $T^N$.
We tested this flow on $\U{3}$ using the target action given in Eq.~\eqref{eq:target_action} with coefficients $c^{(0)}$ from Table~\ref{tab:toy-params} and $\beta=1, 5, 9$. The flow quickly converged with ESS of more than $95\%$ in each case.

\end{document}